\documentclass[11pt,letterpaper]{article}
\usepackage[nottoc]{tocbibind}
\usepackage{amsthm,stmaryrd}
\usepackage{amsmath,amssymb}
\usepackage{url}
\usepackage{comment}
\usepackage[square, numbers]{natbib}
\usepackage{fullpage}
\usepackage{microtype}
\usepackage{MnSymbol}
\usepackage{mathtools}
\usepackage[nice]{nicefrac}
\usepackage{alltt}
\usepackage{float}
\usepackage{graphicx}
\usepackage[ruled,linesnumbered]{algorithm2e}

\SetAlFnt{\small}
\SetAlCapFnt{\small}
\SetAlCapNameFnt{\small}
\SetAlCapHSkip{0pt}
\IncMargin{-\parindent}

\usepackage{color}
\definecolor{DarkBlue}{RGB}{0,0,150}
\usepackage[hidelinks]{hyperref}

\theoremstyle{definition}
\newtheorem{theorem}{Theorem}[section]
\newtheorem*{theorem*}{Theorem}
\newtheorem{definition}[theorem]{Definition}
\newtheorem*{definition*}{Definition}
\newtheorem{lemma}[theorem]{Lemma}
\newtheorem{corollary}[theorem]{Corollary}

\newtheorem*{goal*}{Goal}

\newtheorem{remark}[theorem]{Remark}

\title{On the Power and Limits of Dynamic Pricing in Combinatorial Markets\thanks{
The work of Ben Berger and Michal Feldman is supported by the European Research Council (ERC) under the European Union's Horizon 2020 research and innovation program (grant agreement No. 866132), and by the Israel Science Foundation (grant number 317/17).  The work of Alon Eden is supported by the National Science Foundation under Grant No. CCF-1718549.
}
}

\author{Ben Berger	\\ Tel-Aviv University	\\ \texttt{benberger1@tauex.tau.ac.il}
	\and
		Alon Eden   \\ Harvard University	\\ \texttt{ aloneden@seas.harvard.edu}
	\and
		Michal Feldman 	\\Tel-Aviv University and Microsoft Research \\ \texttt{michal.feldman@cs.tau.ac.il}}  
\date{}

\begin{document}
	\maketitle

\begin{abstract}
	
We study the power and limits of {\em optimal dynamic pricing} in combinatorial markets; i.e., dynamic pricing that leads to optimal social welfare. 
Previous work by Cohen-Addad {\em et al.} [EC'16] demonstrated the existence of optimal dynamic prices for unit-demand buyers, and showed a market with coverage valuations that admits no such prices. 
However, finding the most general class of markets (i.e., valuation functions) that admit optimal dynamic prices remains an open problem.
In this work we establish positive and negative results that narrow the existing gap. 

On the positive side, we provide tools for handling markets beyond unit-demand valuations.  In particular, we characterize all optimal allocations in multi-demand markets.
This characterization allows us to partition the items into equivalence classes according to the role they play in achieving optimality.
Using these tools, we provide a poly-time optimal dynamic pricing algorithm for up to $3$ multi-demand buyers. 
	
On the negative side, we establish a maximal domain theorem, showing that for every non-gross substitutes valuation, there exist unit-demand valuations such that adding them yields a market that does not admit an optimal dynamic pricing. 
This result is the dynamic pricing equivalent of the seminal maximal domain theorem by Gul and Stacchetti [JET'99] for Walrasian equilibrium.
Yang [JET'17] discovered an error in their original proof, and established a different, incomparable version of their maximal domain theorem. 
En route to our maximal domain theorem for optimal dynamic pricing, we provide the first complete proof of the original theorem by Gul and Stacchetti.   
\end{abstract}

\allowdisplaybreaks

\section{Introduction}
\label{sec:intro}

We study the power and limitations of pricing schemes for social welfare optimization in combinatorial markets. 
We consider combinatorial markets with $m$ heterogeneous, indivisible goods, and $n$ buyers with publicly known valuation function $v_i: 2^{[m]}\rightarrow \mathbb{R}_{\geq 0}$ over bundles of items. The goal is to allocate items to buyers in a way that maximizes the social welfare.

Apart from being simple, pricing schemes are attractive since they do not require an all powerful central authority.
Once the prices are set, the buyers arrive and simply choose a desired set of items from the available inventory.  
This is the mechanism we see everywhere, from supermarkets to online stores.
Formally, the seller sets items prices $\mathbf{p}=(p_1, \ldots, p_m)\in \mathbb{R}^m_{\geq 0}$, buyers arrive sequentially in an arbitrary order, and every buyer chooses a bundle $T$ from the remaining items that maximizes the utility: $u_i(T,\mathbf{p})=v_i(T)-\sum_{j\in T}p_j$, breaking ties arbitrarily.

A reader familiar with the fundamental notion of Walrasian equilibrium, which dates back to the 19th century \cite{walras1896elements} (also known as market/pricing/competitive equilibrium), may conclude that the problem is solved for any market that admits a Walrasian equilibrium. A Walrasian equilibrium is a pair of an allocation $\mathbf{S}=(S_1,\ldots,S_n)$ and prices $\mathbf{p}=(p_1,\ldots,p_m)$, such that for every buyer $i$, $S_i$ maximizes $i$'s utility given $\mathbf{p}$. By the first welfare theorem,
every Walrasian equilibrium maximizes social welfare. 

Are Walrasian prices a solution to our problem? The answer, as was previously observed \cite{DBLP:conf/sigecom/Cohen-AddadEFF16,HsuMRRV16}, is no. Walrasian prices alone cannot resolve a market without coordinating the tie breaking. If a buyer is faced with multiple utility-maximizing bundles, it is crucial that a central authority coordinate the tie breaking in accordance with the corresponding optimal allocation. In real-world markets, however, buyers are only faced with the prices and choose a desired bundle by themselves without caring about global efficiency.  \cite{DBLP:conf/sigecom/Cohen-AddadEFF16} demonstrated that lacking a tie-breaking coordinator, Walrasian pricing can lead to an arbitrarily bad allocation. Moreover, they showed that no fixed prices whatsoever can achieve more than $2/3$ of the optimal social welfare in the worst case, even when restricted to unit-demand buyers\footnote{A unit demand buyer has a value for every item, and the value for a set is the maximum value of any item in the set.}.

In order to circumvent this state of affairs, \cite{DBLP:conf/sigecom/Cohen-AddadEFF16} proposed a more powerful pricing scheme, namely {\em dynamic pricing}, in which the seller updates prices in between buyer arrivals. The updated prices are set based on the remaining buyers and the current inventory.
The main result of \cite{DBLP:conf/sigecom/Cohen-AddadEFF16} is that every unit-demand market admits an optimal dynamic pricing.  However, they also showed an example of a market with coverage valuations (a strict sub-class of submodular valuations) in which dynamic prices cannot guarantee optimal welfare. A natural question arises:
\begin{quote}
\textit{What markets (i.e., what valuation classes) can be resolved optimally using dynamic pricing?
}\end{quote}

A similar question was considered for Walrasian equilibrium, where it was shown that every market with gross-substitutes\footnote{Gross-substitutes valuations are a strict super-set of unit-demand.  See formal definition in Section~\ref{sec:prelim}.} buyers admits a Walrasian equilibrium~\cite{kelso1982job}. Moreover, \cite{GS99} show that gross-substitutes valuations are also maximal with respect to guaranteed existence of a Walrasian equilibrium in the following sense:
\begin{theorem}[Maximal Domain Theorem for Walrasian Equilibrium \cite{GS99}] \label{gs_thm}
	Let $v_1$ be a non gross-substitutes valuation.  Then, there exist unit-demand valuations $v_2,\ldots,v_\ell$ for some $\ell$ such that the valuation profile $(v_1,v_2,\ldots,v_\ell)$ does not admit a Walrasian equilibrium\footnote{Yang \cite{Yang17} discovered an error in the proof; details to follow.}.
\end{theorem}

 Although the notions of dynamic pricing and Walrasian are incomparable (see discussion in the last paragraph of our Results and Techniques section), Cohen-Addad {\it et al.} \cite{DBLP:conf/sigecom/Cohen-AddadEFF16} conjectured that GS valuations are also maximal and sufficient with respect to the existence of dynamic prices. For the special case of markets with a unique optimal allocation, they showed that every GS market admits static prices guaranteeing optimal welfare, and there exists a market with non-GS (though submodular) valuations such that no pricing, even dynamic, guarantee optimal welfare.



\subsection{Our Results and Techniques}
In this work we shrink the known gap between markets that can and cannot be resolved optimally via dynamic pricing, from both ends.

A natural extension of unit-demand valuations is {\em multi-demand} valuations, where every buyer $i$ has a public {\em cap} $k_i\in \mathbb{N}$ on the number of desired items, and the value for a set is the sum of the values for the $k_i$ most valued items in the set. The case of $k_i = 1$ is simply unit-demand. Every multi-demand valuation is gross-substitutes. Our main positive result is the following:

\begin{theorem}
	\label{thm:3-md-players}
	Every market with up to $3$ buyers, each with a multi-demand valuation function, admits an optimal dynamic pricing.
	Moreover, the prices can be computed in polynomial (in the number of items $m$) time, using value queries\footnote{A value query for a valuation $v$ receives a set $S$ as input, and returns $v(S)$.}.
\end{theorem}

On the negative side, we show the first general negative result for dynamic prices, which takes the form of a maximal domain result in the spirit of Gul and Stacchetti:

\begin{theorem}\label{maximal_domain_theorem_dynamic_pricing}
	Let $v_1$ be a non gross-substitutes valuation. Then, there are unit-demand valuations $v_2,\ldots,v_\ell$ such that the valuation profile $(v_1,v_2,\ldots,v_\ell)$ does not admit an optimal dynamic pricing.
\end{theorem}

En route, we provide \textit{the first complete proof of the maximal domain theorem by Gul-Stacchetti} (Theorem \ref{gs_thm} above), whose original proof was imprecise (details to follow).  In the following we give an overview of the techniques and tools used to prove these results.

\vspace{-0.1in}
\subsubsection{Techniques: Positive Results.}
%
We first review the solution of Cohen-Addad {\it et al.}~\cite{DBLP:conf/sigecom/Cohen-AddadEFF16} for unit-demand valuations, and show why we need a more fundamental
technique 
in order to get past unit-demand bidders.  In a nutshell, their scheme computes an optimal allocation $\mathbf{X} = (x_1,\ldots,x_n)$, where item $x_i$ is allocated to buyer $i$, and then constructs a {\em complete}, weighted directed graph in which the vertices are the items.  An edge $x_i \rightarrow x_j$ in this graph represents a \emph{preference constraint}, requiring that buyer $i$ {\em strongly} prefer the item $x_i$ over $x_j$, relative to the output prices.  Hereafter, we term this graph the \textit{preference graph}.

If there exist prices $\mathbf{p}$ that satisfy all edge constraints, then all buyers strongly prefer their items over the rest, and the allocation obtained after the last buyer leaves the market is precisely $\mathbf{X}$, which is optimal. Unfortunately, in some cases such prices do not exist. In order to circumvent this problem, \cite{DBLP:conf/sigecom/Cohen-AddadEFF16} proves the following two claims:
\begin{itemize}
	\item An edge $x_i \rightarrow x_j$ participates in a 0-weight cycle iff there is an alternative optimal allocation in which $x_j$ is allocated to buyer $i$.
	\item If 0-weight cycles are removed from the graph, then one can compute prices that satisfy the remaining edge constraints. 
\end{itemize}

Their pricing scheme removes every edge that participates in a 0-weight cycle, and then computes the prices as per the second bullet above.  Relative to these prices, every buyer strongly prefers her allocated item to every other item, except perhaps for the set of items that are allocated to her in some alternative optimal allocation.  Since every buyer takes at most one favorite item, as the buyers are unit-demand, this property guarantees that allocating this item to the buyer is consistent with an optimal allocation (not necessarily $\mathbf{X}$), as desired. When agents are multi-demand, they might take multiple items, and this breaks the solution by \cite{DBLP:conf/sigecom/Cohen-AddadEFF16}. 

\begin{figure}
	\centering\scalebox{0.4}{
		\includegraphics[width=\linewidth]{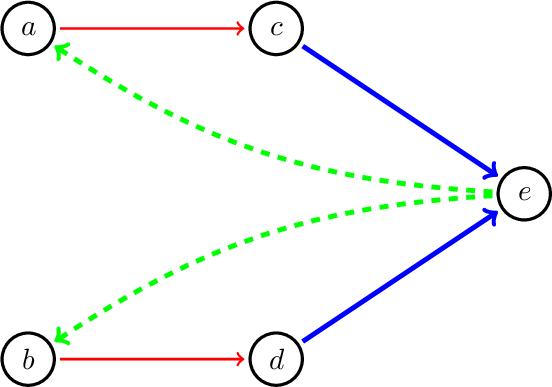}
	}
	
	\caption{Consider a market with 5 items $a,b,c,d,e$ and 3 buyers, 1, 2, 3.  Buyers 1 and 2 are both $2$-demand, and buyer 3 is unit-demand.  Buyer 1 values $a,b,c,d$ at 1 and $e$ at 0, Buyer 2 values $c,d,e$ at 1 and $a,b$ at 0, and buyer 3 values $a,b,e$ at 1 and $c,d$ at 0.  One can verify that allocating $a,b$ to 1, $c,d$ to 2 and $e$ to 3 maximizes social welfare.  The figure depicts two 0-weight cycles in the preference graph constructed in the running example (edge weights are omitted). The thin red, thick blue and dashed green arrows correspond to the constraints of buyers 1,2,3 respectively.}
	\vspace{-10pt}
	\label{fig_simple_bottleneck}
\end{figure}

To illustrate this, consider the example given in Figure \ref{fig_simple_bottleneck}, which serves as a running example throughout the paper.  Removing the given 0-weight cycles could result in buyer 1 taking $c$ and $d$ instead of $a$ and $b$, and the only remaining item that gives buyer 2 any positive value is $e$.  This decreases the maximum attainable welfare from 5 to 4. The reason for this is that the two cycles intersect, and item $e$ acts as a bottleneck for the two cycles. The machinery developed in \cite{DBLP:conf/sigecom/Cohen-AddadEFF16} cannot identify the special role of item $e$, which is crucial for resolving this instance.

Our first step is to gain a better structural understanding of optimal allocations in multi-demand markets. This is cast in  the following theorem that characterizes the set of optimal allocations in multi-demand markets with any number of buyers.  For the sake of simplicity, we present the theorem for markets in which all $m$ items are allocated in every optimal allocation, and in which the total demand of the players equals supply, i.e. $m = \sum_{i=1}^{n}k_i$, where $k_i$ is the cap of buyer $i$. An analogous result holds in the general case (see Appendix \ref{appendix: m < sum k_i}).

\begin{theorem*}[Informal. See Theorem \ref{legal_iff_optimal}]
	In a market with multi-demand buyers, an allocation is optimal if and only if the following hold:
	\vspace{-0.1in}
	\begin{itemize}
		\item Every buyer $i$ receives $k_i$ items.
		\item If item $x$ is allocated to buyer $i$, then there exists an optimal allocation where $x$ is allocated to $i$. 
	\end{itemize}
\end{theorem*}

Put informally, the above states that one can mix-and-match items given to a buyer in \textit{different} optimal allocations, and as long as each buyer $i$ receives {\em exactly} $k_i$ items, the resulting allocation is also optimal. While the only if direction is straightforward, it is not a-priori clear that the other direction holds as well.  
%
We prove this direction by reducing the problem to unit-demand valuations and proving for this case. 

This characterization significantly simplifies the problem. 
It allows us to ignore the concrete values, and consider for each item only the set of buyers that receives it in some optimal allocation.
Two items are essentially equivalent if their corresponding sets of buyers coincide. Thus, we group items into {\em equivalence classes}, providing a compact view of the market.  For example, in markets with up to 3 multi-demand buyers, there are at most 8 (non-empty) equivalence classes corresponding to the possible subsets of players, while the total number of items can be arbitrarily large. We construct a new directed graph, termed the {\em item-equivalence graph}, where the vertices are these equivalence classes (refined after intersecting them with the bundles from the initial optimal allocation $\mathbf{X}$), and there is an edge $C \rightarrow D$ whenever the buyer that receives the items in $C$ in the allocation $\mathbf{X}$ also receives every item in $D$ in some optimal allocation.  Figure \ref{bottleneck_example1} depicts the item-equivalence graph for the running example.

\begin{figure}
	\centering\scalebox{0.4}{
		\includegraphics[width=\linewidth]{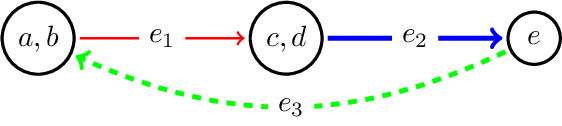}
	}
	\caption{The item-equivalence graph for the running example. E.g., the items $a,b$ are equivalent since the set of buyers that receive any of them in some optimal allocation is the same ($\{1,3\}$).}
	\vspace{-10pt}
	\label{bottleneck_example1}
\end{figure}

We show that there is a correspondence between cycles in the item-equivalence graph and 0-weight cycles in the preference graph.   Thus our challenge is reduced to removing enough edges from the first (and translating these removals back to the second), in a way that eliminates all cycles, but also guarantees the following:  every deviation by any buyer from her prescribed bundle, implied by the edge removals, allows the other buyers to simultaneously compensate for their ``stolen'' items by replacing them with items from other relevant equivalence classes.  The optimal allocation characterization theorem then guarantees that the obtained allocation is indeed optimal.  We devise an edge-removal method satisfying these requirements whenever the number of buyers is at most 3.

We believe this characterization 
theorem and the item equivalence graph may prove useful in other problems related to multi-demand markets.

%

\subsubsection{Techniques: Negative Results.}
%

The original proof of Theorem~\ref{gs_thm} by Gul and Stacchetti considers two cases, and for each case, they construct a different market that does not admit a Walrasian equilibrium.
However, as was observed by Yang~\cite{Yang17}, one of the constructions does not work. Indeed, Yang shows an instance such that the constructed market does admit a Walrasian equilibrium. 
Unfortunately, the error could not be easily fixed, and Yang proceeded by establishing an alternative, incomparable theorem; namely, that for every non gross-substitutes valuation there is a (single) gross-substitutes valuation for which the obtained market has no Walrasian equilibrium. While Yang's version of the assertion requires only a single valuation, this valuation has a complex structure (compared with the simple unit-demand valuations in the original version).  
In Section \ref{sec: MDT} we prove the maximal domain theorem \textit{as it was originally stated} (see Theorem \ref{gs_thm}). The proof relies on a theorem  (see Theorem \ref{B_minus_A_2}) which allows us to consider only the case with the correct construction in the original proof. 


Our proof of Theorem~\ref{maximal_domain_theorem_dynamic_pricing} is driven by the following lemma from Cohen-Addad {\it et al.}~\cite{DBLP:conf/sigecom/Cohen-AddadEFF16} --- in the case of a unique optimal allocation, the existence of optimal dynamic prices implies the existence of Walrasian prices. Once correcting the proof of Gul and Stacchetti, we modify the \textit{correct} construction in their proof to a market with an optimal allocation that is ``almost" unique, which still does not admit a Walrasian equilbrium. We then adapt the lemma in~\cite{DBLP:conf/sigecom/Cohen-AddadEFF16} to show that the existence of optimal dynamic prices in this market also implies the existence of Walrasian prices. The non-existence of Walrasian prices now implies non-existence of optimal dynamic prices.
%

We stress that for proving our maximal domain theorem for dynamic prices, it does not suffice to provide a correct proof of the Gul-Stacchetti theorem, as optimal dynamic prices \textit{do not} imply Walrasian prices. Indeed, in Appendix \ref{appendix:example_dp_no_we} we show an instance that admits optimal dynamic prices but not Walrasian prices.
Interestingly, it was already observed by~\cite{DBLP:conf/sigecom/Cohen-AddadEFF16} that the existence of Walrasian prices does not imply the existence of optimal dynamic prices. Putting these together implies the notions are incomparable.

\vspace{0.1in}
\noindent {\bf Open problems.}  Our results suggest questions for future research. The most obvious one is whether our positive result for 3 multi-demand buyers can be extended to any number of buyers. Recall that some of the tools developed in this work (such as the optimal allocation characterization theorem, and the item equivalence graph) are applicable to multi-demand markets of any size. Thus, they may prove useful in extending our positive result beyond 3 buyers. More generally, it is still open whether any market with gross substitutes valuations admits an optimal dynamic pricing.

\subsection{Related Work}
The notion of Walrasian equilibrium was defined for divisible-goods as early as the 19th century~\cite{walras1896elements}. This notion was later extended to combinatorial markets, where Kelso and Crawford~\cite{kelso1982job} introduced the class of gross-substitutes valuations as a class for which a natural ascending auction reaches a Walrasian equilibrium. Gul and Stacchetti later showed via their maximal domain theorem that gross-substitutes is the frontier of this existence result~\cite{GS99,Yang17}. Gross-substitutes valuations have been introduced independently in different fields, under different names, and under seemingly different definitions~\cite{dress1995rewarding,dress1995well,dress1990valuated,murota1999m}; see \cite{Leme17} for a comprehensive survey of gross-substitutes valuations. In order to circumvent the non-existence of a market equilibrium under broader valuation classes, relaxations of market equilibrium were introduced~\cite{FeldmanGL16,DughmiEFFL16}, and behavioral biases were harnessed~\cite{BabaioffDO18,EzraFF19}.

Posted price mechanisms were shown to be useful in combinatorial markets. Feldman, Gravin and Lucier~\cite{FeldmanGL15} showed how to compute simple ``balanced" static prices in order to obtain at least half of the optimal welfare for submodular valuations, even in the case where the seller has only Bayesian knowledge about the valuations. This idea was generalized by~\cite{DuettingFKL17}, and was shown to be useful even in the face of complementarities between items~\cite{ChawlaMT19}. 

\cite{DBLP:conf/sigecom/Cohen-AddadEFF16} and \cite{HsuMRRV16} were the first to demonstrate that Walrasian prices cannot even approximate the optimal welfare in the absence of a centralized tie-breaking coordinator. Cohen-Addad {\em et al.} resolved this issue by adjusting prices dynamically for unit-demand valuations. They also showed an instance of coverage valuations where a Walrasian equilibrium exists and yet dynamic prices cannot guarantee optimal welfare. On the other hand, Hsu et al. showed that under some conditions, minimal Walrasian prices guarantee near-optimal welfare for a strict subclass of gross-substitutes valuations\footnote{\cite{tran2019finite} recently showed that the class of matroid-based valuations is a strict subclass of gross-substitutes valuations.}. \cite{EzraFRS18} and \cite{EdenFF19} established better guarantees via static pricing for simpler markets (identical items and binary unit-demand, respectively), in comparison to \cite{FeldmanGL15}. 


Posted price mechanisms have been shown to be useful in additional settings with different objective functions, including revenue maximization in combinatorial markets~\cite{GuruswamiHKKKM05},\cite{ChawlaHK07},\cite{ChawlaHMS10},\cite{AnshelevichS17}, cost minimization in online scheduling~\cite{FeldmanFR17},\cite{EdenFFT18},\cite{ImMPS17}, and a variety of other online resource allocation problems~\cite{CohenEFJ15},\cite{CohenEFJ19},\cite{FiatMN08},\cite{AzarBFFS19}.



\section{Preliminaries}
\label{sec:prelim}

We consider a setting with a finite set of indivisible items $M$ (with $m := \left|M\right|$) and a set of $n$ buyers (or players).  Every buyer has a valuation function $v: 2^M \rightarrow \mathbb{R}_{\geq 0}$. As standard, we assume monotonicity and normalization of all valuations, i.e. $v(S)\leq v(T)$ whenever $S\subseteq T$, and $v(\emptyset) =0$.
A valuation profile of $n$ buyers is denoted $\mathbf{v}=(v_1,\ldots,v_n)$ and we assume that it is known by all.  An \textit{allocation} is a vector $\mathbf{A}=(A_1,\ldots, A_n)$ of disjoint subsets of $M$, indicating the bundles of items given to each player (not all items have to be allocated).  The \textit{social welfare} of an allocation $\mathbf{A}$ is given by $\mathsf{SW}(\mathbf{A}) = \sum_{i=1}^{n}v_i(A_i)$.  An \textit{optimal allocation} is an allocation that achieves the maximum social welfare among all allocations.

A \textit{pricing} or a \textit{price vector} is a vector $\mathbf{p} \in \mathbb{R}_{\geq 0}^m$ indicating the price of each item.  We assume a quasi-linear utility, i.e. the \textit{utility} of a buyer $i$ from a bundle $S\subseteq M$ given prices $\mathbf{p}$ is $u_i(S,\mathbf{p}) = v_i(S)-\sum_{x \in S} p_x$.  The \textit{demand correspondence} of buyer $i$ given $\mathbf{p}$ is the collection of utility maximizing bundles $D_{\mathbf{p}}(v_i) := \arg \max_{S\subseteq M} \{u_i (S,\mathbf{p})\}$.

\vspace{0.1in}
\noindent {\bf Dynamic Pricing.}
In the dynamic pricing problem buyers arrive to the market in an arbitrary and unknown order.  Before every buyer arrival new prices are set to the items that are still available, and these prices are based only on the set of buyers that have not yet arrived (their arrival order remains unknown).  The arriving buyer then chooses an arbitrary utility-maximizing bundle based on the current prices and available items.  The goal is to set the prices so that for any arrival order and any tie breaking choices by the buyers, the obtained social welfare is optimal.

We are interested in proving the guaranteed existence of an optimal dynamic pricing for any market composed entirely of buyers from a given valuation class $C$.  It can be easily shown by induction that the problem is reduced to proving the guaranteed existence of item prices $\mathbf{p}$ such that any utility-maximizing bundle of any buyer can be completed to an optimal allocation.  In other words, we can rephrase optimal dynamic pricing as follows:

\begin{definition}
	\label{def:dynamic-pricing}
	An \textit{optimal dynamic pricing} (hereafter, dynamic pricing) for the buyer profile $\mathbf{v}=(v_1,\ldots,v_n)$ is a price vector $\mathbf{p} \in \mathbb{R}_{\geq 0}^{m}$ such that for any buyer $i$ and any $S\in D_{\mathbf{p}}(v_i)$ there is an optimal allocation $\mathbf{O}$ in which player $i$ receives $S$.
\end{definition}



\vspace{0.1in}
\noindent{\bf Valuation Classes.}
Besides unit-demand and multi-demand valuations, which were presented in the introduction, this paper also considers gross-substitutes valuations.  A valuation $v: 2^M \rightarrow \mathbb{R}_{\geq 0}$ is \textit{gross-substitutes} if for any two price vectors $\mathbf{p},\mathbf{q}$ such that $\mathbf{p} \leq \mathbf{q}$ (point-wise), and for any $A \in D_{\mathbf{p}}(v)$ there is a bundle $B \in D_{\mathbf{q}}(v)$ such that $A\cap \{x\in M \mid p_x = q_x\} \subseteq B$.

%


%

\section{Dynamic Pricing for Multi-Demand Buyers} \label{sec: dynamic_pricing}

In this section we prove Theorem \ref{thm:3-md-players}, namely we establish a 
dynamic pricing scheme for up to $n=3$ multi-demand buyers that runs in $poly(m)$ time.  As we shall see, most of the tools we use hold for any number of buyers $n$.  We fix a multi-demand buyer profile $\mathbf{v} = (v_1,\ldots, v_n)$ over the item set $M$, where each $v_i$ is $k_i$-demand. We assume w.l.o.g. that all items are essential for optimality (i.e. all items are allocated in every optimal allocation) since otherwise we can price all unnecessary items at $\infty$ in every round to ensure that no player takes any of them (and price the rest of the items as if the unnecessary items do not exist).  Note that under this assumption, each optimal allocation gives buyer $i$ at most $k_i$ items, for every $i$.  In particular we have $m\leq \sum_{i=1}^{n} k_i$.  For the sake of simplicity we further assume for the rest of this section that every optimal allocation gives each buyer $i$ exactly $k_i$ items, and thus $m=\sum_{i=1}^{n}k_i$.  The case $m < \sum k_i$ introduces substantial technical difficulty and we defer its treatment to Appendix \ref{appendix: m < sum k_i}.  
We first go over the tools used in our dynamic pricing scheme.  With the tools at hand, we present the dynamic pricing scheme for $n=3$ buyers.

\subsection{Tools and Previous Solutions}\label{sec: tools}
We start by presenting
the main combinatorial construct of our solution, namely the preference-graph, which generalizes the construct given by \cite{DBLP:conf/sigecom/Cohen-AddadEFF16} in their solution for unit-demand buyers.  
Then
we explain the obstacles for generalizing the approach of \cite{DBLP:conf/sigecom/Cohen-AddadEFF16} to the multi-demand setting.
Finally,
we develop the necessary machinery needed to overcome these obstacles.  All the tools we develop and their properties hold for any number of buyers $n$.

\vspace{0.1in}
\noindent {\bf The Preference Graph and an Initial Pricing Attempt.}
Let $\mathbf{O}$ be an arbitrary optimal allocation.  The preference graph based on $\mathbf{O}$ is the directed graph $H$ whose vertices are the items in $M$.  Furthermore there is a special `source' vertex denoted $s$. For any two different players $i,j$ and items $x\in O_i,y\in O_j$ we have a directed edge $e=x\rightarrow y$ with weight $w(e)= v_{i}(x)-v_{i}(y)$.  We also have a 0-weight edge $s\rightarrow x$ for every item $x\in M$.  Since an optimal allocation can be computed in $poly(n,m)$ time with value queries (since the valuations are gross substitutes, see \cite{Leme17}), it follows that the preference graph can also be computed in $poly(n,m)$ time with value queries.
When $\left|O_i\right| = 1$ for every $i$, the graph is exactly the one introduced by \cite{DBLP:conf/sigecom/Cohen-AddadEFF16} in their unit-demand solution\footnote{A similar graph structure has been used by Murota in order to compute Walrasian equilibria in gross-substitutes markets (\cite{Murota96}).}.
The proofs of the following two lemmas and corollary are deferred to Appendix \ref{proofs: section dynamic_pricing}.

\begin{lemma} \label{pass_alloc}
	Let $C := x_1 \rightarrow x_2 \rightarrow \cdots x_k \rightarrow x_1$ be a cycle in $H$, where $x_i$ is allocated to player $i$ in $\mathbf{O}$ and $x_i \neq x_j$ for every $i \neq j$.  Then the weight of the cycle is $w(C)=\mathsf{SW}(\mathbf{O}) - \mathsf{SW}(\mathbf{A})$
	where $\mathbf{A}$ is the allocation obtained from $\mathbf{O}$ by transferring $x_{i+1}$ to player $i$ for every $i$ (we identify player $k+1$ with player $1$).
\end{lemma}

\begin{corollary} \label{non_negative_cycles}
	Every cycle in $H$ has non-negative weight.
\end{corollary}

Corollary \ref{non_negative_cycles} implies that the weight of the min-weight path from $s$ to $x$, denoted $\delta (s,x)$, is well-defined for any item $x$.

\begin{lemma} \label{non_negative_prices_util_dif}
	Let $p_x := -\delta (s,x)$ for every item $x$.  Let $i$ be some player, and let $x,y$ be items such that $x \in O_i, y \notin O_i$. Then:
	\begin{enumerate}
		\item $p_x \geq 0$.
		\item $v_{i}(x) - p_x \geq v_{i}(y) - p_y$.
		\item $v_{i}(x) - p_x \geq 0$.
	\end{enumerate}
\end{lemma}

Note that the utility player $i$ obtains from any bundle of size at most $k_i$ is the sum of the individual utilities obtained by the individual items.  
Thus, Lemma \ref{non_negative_prices_util_dif} shows that setting the prices $p_x = -\delta(s,x)$ almost achieves the requirements of dynamic pricing.  
However, since the inequalities in Lemma \ref{non_negative_prices_util_dif} are not strict, the incoming player might deviate from the designated bundle.  



\vspace{0.1in}
\noindent {\bf Solution for Unit-Demand Valuations and its Failure to Generalize.}
The inequalities of Lemmas~\ref{non_negative_prices_util_dif} can be made strict by decreasing the weight of all edges by an appropriately selected $\varepsilon > 0$, but in the case $H$ has zero-weight cycles, this can introduce negative cycles to $H$, in which case $\delta(s,x)$ is not defined for any $x$ in such cycle. To circumvent this issue,~\cite{DBLP:conf/sigecom/Cohen-AddadEFF16} remove every edge that participates in a 0-weight cycle in $H$. Therefore, by choosing a small enough $\epsilon$ to decrease from the remaining edges, the remaining cycles are guaranteed to be strictly positive. 
Removing an edge $x \rightarrow y$ for $x \in O_i$, $y \notin O_i$ cancels the preference guarantee of Lemma \ref{non_negative_prices_util_dif} (part 2), leading to a possible deviation by buyer $i$ from taking $x$ to taking $y$. However, since 0-weight cycles correspond to alternative optimal allocations (see Lemma \ref{pass_alloc} with $w(C)=0$), then this is not a problem:  if the edge $x \rightarrow y$ was removed, then there is an optimal allocation in which player $i$ receives $y$ instead of $x$.  As for the edges $x\rightarrow y$ that were not removed, the $\epsilon$ decrement causes $i$ to strongly prefer $x$ over $y$.  The other inequalities of Lemma \ref{non_negative_prices_util_dif} would also be strict, and we are thus guaranteed that the incoming player indeed takes a one-item bundle that is part of some optimal allocation, as desired.  

This approach works in the unit-demand setting, but poses problems in the multi-demand setting, as illustrated in the running example (presented in the introduction, see Figure \ref{fig_simple_bottleneck}).
%
%
The example shows that
a more sophisticated method of eliminating 0-weight cycles must be employed instead of simply removing all edges that participate in some 0-weight cycle.  To be more precise, we state our informal goal:  
\begin{quote}
{\it
Remove a set of edges from the preference graph so that no 0-weight cycles are left, and every possible deviation implied by the removed edges is consistent with some optimal allocation.
}
\end{quote}



\vspace{0.1in}
\noindent {\bf Legal Allocations.}
\begin{definition} $ $
	\begin{itemize} 
		\item An item $x \in M$ is \textit{legal} for player $i$ if there is some optimal allocation $\mathbf{X} = (X_1, \ldots, X_n)$ such that $x \in X_i$.
		\item A bundle $S \subseteq M$ is \textit{legal} for player $i$ if $|S| = k_i$ and every $x \in S$ is legal for player $i$.
		\item A \textit{legal allocation} $\mathbf{A} = (A_1, \ldots, A_n)$ is an allocation in which $A_i$ is legal for player $i$, for every $i$.  
	\end{itemize}
\end{definition}

In a legal allocation every player $i$ receives exactly $k_i$ items, each of which is allocated to her in some optimal allocation. Note that a legal bundle for buyer $i$ might not form a part of any optimal allocation (e.g., the bundle $\{c,d\}$ for buyer 1 in the running example).
The following theorem provides a characterization of the collection of optimal allocations in the given market $\mathbf{v}$.  The subsequent Corollary follows directly from the theorem and Definition \ref{def:dynamic-pricing}.

\begin{theorem}
	\label{legal_iff_optimal} An allocation is legal if and only if it is optimal.
\end{theorem}

\begin{corollary}\label{main_theorem}
	A price vector $\mathbf{p}$ is a dynamic pricing if for every player $i$ and $S \in D_{\mathbf{p}}(i)$, $S$ is legal for player $i$ and there exists an allocation of the items $M\setminus S$ to the other players in which every player receives a bundle that is legal for her.
\end{corollary} 
Thus, going back to our informal goal, Theorem \ref{legal_iff_optimal} determines the deviations from the bundles $O_i$ which are tolerable.  A buyer can only deviate to a bundle which is legal for her, in a way that the leftover items can be partitioned ``legally'' among the rest of the buyers.
We now prove Theorem \ref{legal_iff_optimal}.

\begin{proof}
	Legality follows from optimality due to our assumption that every optimal allocation allocates exactly $k_i$ items to every player $i$.  For the other direction we first prove the theorem in the special case of a unit-demand market, i.e. $k_i = 1$ for every player $i$, and then we show how the general case reduces to the unit-demand market case.
	\begin{lemma} \label{matching_markets}
		If $k_i = 1$ for all $i$ then any legal allocation is optimal.
	\end{lemma}
	\begin{proof}
		
		Let $\mathbf{L} = (\{\ell_1\},\ldots, \{\ell_n\})$ be a legal allocation, and let $\mathbf{O} = (\{o_1\},\ldots, \{o_n\})$ be some optimal allocation.	We show that $\mathsf{SW}(\mathbf{L}) = \mathsf{SW}(\mathbf{O})$. By definition, for every $i$ there is an optimal allocation $\mathbf{O}^i$ in which the item $\ell_i$ is allocated to player $i$ and optimality implies that $\mathsf{SW}(\mathbf{O}^i) = \mathsf{SW}(\mathbf{O})$.  Summing over all $i$ we get
		\begin{align}
		\sum_{i=1}^n \mathsf{SW}(\mathbf{O}) = \sum_{i=1}^{n} \mathsf{SW}(\mathbf{O}^i) \label{summing}
		\end{align}
		
		The right side in (\ref{summing}) accounts for the welfare of $n$ optimal allocations, each of which is by itself a sum of $n$ terms of the form $v_j(x)$ where every item and player appears exactly once (recall the assumption that every optimal allocation allocates all items, i.e., is a perfect matching).  Thus, the right side in (\ref{summing}) is a sum of $n^2$ terms in which every item and player appears exactly $n$ times.  Consider the bipartite graph $G=(X,Y;E)$ in which the left side $X$ is the set of items,  the right side $Y$ is the set of players and there is an edge $\{x,j\}$ for each summand $v_j(x)$ appearing in the right side of (\ref{summing}).  Then $G$ is an $n$-regular bipartite graph (possibly with multi-edges), and note that the (perfect) matching induced by the legal allocation $\mathbf{L}$ appears in $G$ (since the term $v_i(\ell_i)$ appears in $\mathsf{SW}(\mathbf{O}^i)$, for every i). If we erase the edges of that matching from $G$ then we are left with an $(n-1)$-regular bipartite graph.  It is a well-known fact that in this case the edges of $G$ can be split into $n-1$ perfect matchings $\mathbf{P}^1,\ldots, \mathbf{P}^{n-1}$.  Thinking of these matchings as allocations, and together with the allocation $\mathbf{L}$, we get
		\[
		\sum_{i=1}^n \mathsf{SW}(\mathbf{O}) = \sum_{i=1}^{n} \mathsf{SW}(\mathbf{O}^i) = \mathsf{SW}(\mathbf{L}) + \sum_{i=1}^{n-1} \mathsf{SW}(\mathbf{P}^i).
		\] 
		Since $\mathsf{SW}(\mathbf{O}) \geq \mathsf{SW}(\mathbf{L})$ and $\mathsf{SW}(\mathbf{O}) \geq \mathsf{SW}(\mathbf{P}^i)$ for every $i$ (by optimality of $\mathbf{O}$), it must be the case that all weak inequalities are in fact equalities, establishing in particular that $\mathsf{SW}(\mathbf{O}) = \mathsf{SW}(\mathbf{L})$, as desired.
	\end{proof}

	We now describe the reduction from the general case to the unit-demand market case. We are given a market with $n$ agents where each agent is $k_i$-demand and $\sum_i k_i=m$.  Our reduction keeps the same set of items $M$, but splits each agent $i$ to $k_i$ identical unit-demand agents, where for each copy, the value of the agent for an item $x\in M$ is simply $v_i(x)$. Clearly, the number of unit-demand agents as a result of this reduction is $\sum_i k_i = m$. Let $N_i$ be the set of unit-demand bidders that corresponds to agent $i$, and let $N'=\bigcup_{i\in N} N_i$.
	
	Given an allocation $\mathbf{O}=(O_i)_{i\in N}$ for the original market, where $|O_i|= k_i$ for every $i$, the corresponding allocation in the unit-demand market splits the $|O_i|$ different items arbitrarily between the unit-demand bidders corresponding to bidder $i$, with each bidder receiving exactly one item. Notice that the social welfare achieved by this allocation is the same as in the original allocation.  Similarly, given an allocation $\mathbf{O'}=\{O'_\ell\}_{\ell\in N'}$ in the unit-demand market, the corresponding allocation in the original market gives all the items allocated to agents in $N_i$ to agent $i$. Again, since $|N_i|=k_i$ the resulting allocation achieves the same welfare as the original one.
	
	\begin{lemma} \label{lem:equivalence}
		An allocation is legal in the original market if and only if it is legal in the corresponding unit-demand market. An allocation is optimal in the original market if and only if it is optimal in the corresponding unit-demand market.
	\end{lemma}
	\begin{proof}
		Consider an optimal allocation in the original market and recall our assumption that optimal allocations give each player $i$ exactly $k_i$ items. The corresponding allocation in the unit-demand market obtains the same value. Similarly, given an optimal allocation in the unit-demand market where each agent gets one item, the corresponding allocation in original market obtains the same value. Therefore, an allocation is optimal in the original market if and only if it is optimal in the corresponding unit-demand market.
		
		Similar reasoning shows that given a legal allocation in the original market, the corresponding allocation in the unit-demand market is legal as well and vice versa.
	\end{proof}
	
	To complete the proof of Theorem \ref{legal_iff_optimal}, take a legal allocation in the original market. According to Lemma~\ref{lem:equivalence} it is also legal in the corresponding unit-demand market. From Lemma \ref{matching_markets}, we get that it is optimal in the unit-demand market. Again, by Lemma~\ref{lem:equivalence}, we get that the corresponding allocation in the original market is also optimal.
\end{proof} 

\vspace{0.1in}
\noindent {\bf The Item-Equivalence Graph.}
Let $\mathbf{O}$ be some optimal allocation and $H$ the corresponding preference graph.  For every player $i$ and set of players $C \subseteq \left[n\right]\setminus \{i\}$, we denote by $B_{i,C}$ the set of items
allocated to buyer $i$ in $\mathbf{O}$, and whose set of players to which they are legal is exactly $\{i\} \cup C$.
For example, $B_{1,\{2,3\}}$ is the set of items $x\in O_1$ such that there are optimal allocations $\mathbf{O}'$,$\tilde{\mathbf{O}}$ in which $x$ is allocated to players 2, 3 (respectively), and for any other player $j \notin \{1,2,3\}$, there is no optimal allocation in which $x$ is allocated to $j$.  Formally,
\begin{eqnarray*}
	B_{i,C}:=\left\{x \in O_i \biggm| 	\begin{matrix*}[l] 	\forall j\in \{i\}\cup C
		&x \text{ is legal for } j \\
		\forall j\notin \{i\}\cup C 
		&x \text{ is not legal for } j \end{matrix*} \right\}
\end{eqnarray*}
We make a few observations:
\begin{itemize}
	\item The sets $B_{i,C}$ form a partition of $M$ (some of these sets might be empty sets).	
	\item Let $x \in O_i$ and $y \in O_j$ for $i \neq j$.  If $x \rightarrow y$ participates in a 0-weight cycle in $H$ and $y \in B_{j,C}$, then $i \in C$.
\end{itemize}

The second observation holds since if $x \rightarrow y$ participates in a 0-weight cycle, then there is an alternative optimal allocation in which $y$ is allocated to player $i$ (see Lemma \ref{pass_alloc} with $w(C) = 0$).

\begin{definition}[Item-Equivalence Graph]
	Given an optimal allocation $\mathbf{O}$, its associated \textit{item-equivalence graph} is the directed graph $B = \left(T,D \right)$ with vertices ~$T=\{B_{i,C}\neq \emptyset \mid i\in\left[n\right], C\subseteq \left[n\right]\setminus\{i\}\}$
	and directed edges $D = \left\{ B_{i, C_1} \rightarrow B_{j, C_2} \mid  i \in C_2 \right\}$.

\end{definition}

For example, $(B_{1,\emptyset} \rightarrow B_{2,\{1,4\}})$ and $(B_{2,\{1,5\}}\rightarrow B_{6,\{2\}})$ are edges in the item-equivalence graph (assuming that the participating sets are non-empty), whereas, for example, $(B_{1,\emptyset} \rightarrow B_{1,\{2\}})$ and $(B_{2,\{1\}} \rightarrow B_{3,\{1,4\}})$ are not.  Note also that the number of vertices is at most $m$.

The next Lemma shows that the item-equivalence graph can be computed efficiently.  The proof is deferred to Appendix \ref{proofs: section dynamic_pricing}.

\begin{lemma}\label{poly_time_item-equivalence_graph}
	Given an optimal allocation $\mathbf{O}$, its associated item-equivalence graph can be computed in $poly(m,n)$ time and value queries.
\end{lemma}

The following lemma uses Theorem \ref{legal_iff_optimal} to establish a correspondence between 0-weight cycles in $H$ and cycles in $B$.  Its proof is deferred to Appendix \ref{proofs: section dynamic_pricing}.

\begin{lemma} \label{item-equivalence_cycles}
	Let $\mathbf{O}$ be an optimal allocation and let $H$ and $B$ be the corresponding preference graph and item-equivalence graph, respectively. Then:
	\begin{enumerate}
		\item If $B_{i_1, C_1} \rightarrow \cdots \rightarrow B_{i_k, C_k} \rightarrow B_{i_1, C_1}$ is a cycle in $B$ then for any items $x_1 \in B_{i_1, C_1} , \ldots, x_k \in B_{i_k, C_k}$, the cycle $C=x_1\rightarrow x_2\rightarrow \cdots \rightarrow x_k\rightarrow x_1$ is a 0-weight cycle in $H$. 
		
		\item If $C=x_1\rightarrow x_2\rightarrow \cdots \rightarrow x_k\rightarrow x_1$ is a 0-weight cycle in $H$, and $x_\ell \in B_{i_\ell, C_\ell}$ for every $1 \leq \ell \leq k$ then $C' := B_{i_1, C_1} \rightarrow \cdots \rightarrow B_{i_k, C_k} \rightarrow B_{i_1, C_1}$ is a cycle in $B$. 
	\end{enumerate}	
\end{lemma}

\begin{proof}
	Let $C = B_{i_1, C_1} \rightarrow \cdots B_{i_k, C_k} \rightarrow B_{i_1, C_1}$ be a cycle in $B$, let $x_1 \in B_{i_1, C_1},\ldots, x_k \in B_{i_k, C_k}$ and note that the edges $x_i \rightarrow x_{i+1}$ exist in $H$ (since these are items that belong to different players).  Consider the cycle $C' = x_1 \rightarrow \cdots \rightarrow \cdots x_k \rightarrow x_1$ in $H$.  We need to show that $w(C')=0$. We can assume w.l.o.g. that all the items $x_\ell$ are different as otherwise $C'$ is a union of two or more cycles for which this assumption holds (these cycles are derived from corresponding sub-cycles of $C$), and the weight of a union of 0-weight cycles is 0.  By definition of $B$ we have $i_{\ell} \in C_{\ell + 1}$ for every $1 \leq \ell \leq k$ (again we identify $k+1$ with $1$) and so we conclude that the allocation $\mathbf{A}$ obtained from $\mathbf{O}$ by passing $x_{\ell + 1}$ to player $i_\ell$ is a legal allocation and thus optimal (by Theorem \ref{legal_iff_optimal}).  Since $\mathbf{O}$ is also optimal, we conclude by Lemma \ref{pass_alloc} that 
	$
	w(C) = \mathsf{SW}(\mathbf{O}) - \mathsf{SW}(\mathbf{A}) = 0.
	$
	We now prove part 2.  Let $C = x_1 \rightarrow \cdots \rightarrow x_k \rightarrow x_1$ be a 0-weight cycle in $H$.  Again we can assume w.l.o.g. that all items are different. For all $\ell$ let $i_\ell$ and $C_\ell$ be the player and set such that $x_\ell \in B_{i_\ell, C_\ell}$.  By Lemma \ref{pass_alloc} the allocation $\mathbf{A}$ obtained from $\mathbf{O}$ by passing $x_{\ell + 1}$ to player $i_\ell$ is optimal and thus legal.  We conclude that $i_{\ell} \in C_{\ell + 1}$ for all $1 \leq \ell \leq k$ implying that the edges $B_{i_\ell, C_{\ell}} \rightarrow B_{i_{\ell + 1}, C_{\ell + 1}}$ exist in $B$. Therefore $C' = B_{i_1, C_1} \rightarrow \cdots B_{i_k, C_k}\rightarrow B_{i_1, C_1}$ is a cycle in $B$. 
\end{proof}


%
As explained before, our main challenge in the dynamic pricing problem is to come up with a method to remove all 0-weight cycles from $H$ in a way that each potential deviation of any player $i$ from the designated bundle $O_i$, that emanates from the edge removals, is consistent with some optimal allocation.  In particular the method must overcome the ``bottleneck problem''  (as illustrated in Figure \ref{fig_simple_bottleneck}).  Lemma \ref{item-equivalence_cycles} allows us to shift the focus from removing 0-weight cycles in $H$ to removing cycles in $B$ and translate these removals back to $H$.

\vspace{0.1in}
\noindent {\bf [Running Example]}
Figure \ref{bottleneck_example1} (presented in the introduction) shows the item-equivalence graph obtained from the initial optimal allocation.  Each of the items $a,b$ is allocated to buyer 3 in some other optimal allocation (and is never allocated to buyer 2).  Thus $a,b \in B_{1,\{3\}}$.  Similarly we have $c,d \in B_{2,\{1\}}$ and $e \in B_{3,\{2\}}$.  Note that removing any of the edges of the item-equivalence graph makes it cycle-free.  Thus, by Lemma \ref{item-equivalence_cycles}, if we choose one of the edges $e_1,e_2,e_3$ and remove all edges in the preference graph corresponding to the chosen edge, then the preference graph will remain cycle-free.  Now, removing the edges corresponding to $e_1$ could cause player 1 to take the bundle $\{c,d\}$ instead of the designated bundle $\{a,b\}$, and this cannot be completed to an optimal allocation.  On the other hand, removing the preference graph edges that correspond to the edges $e_2$ and/or $e_3$ is fine.  If player 2 arrives first to the market, then the removal of edge $e_2$ might cause her to take the item $e$ instead of $c$ or $d$, and both options are consistent with some optimal allocation.  Likewise if player 3 arrives first and takes $a$ or $b$ instead of $e$ then this too can be completed to an optimal allocation.  The important property here is that $B_{3,\{2\}}$ has minimal size in the cycle,
and thus removing its incoming and outgoing edges introduces tolerable potential deviations.
\vspace{0.1in}

%

\subsection{Solution for up to 3 Multi-Demand Buyers.}\label{three_player_solution}

We are now ready to present the dynamic pricing scheme for up to $n=3$ multi-demand buyers.
\begin{remark}
The algorithm makes use of the item-equivalence graph. We abuse notation and instead of writing $B_{i,\{j\}}$ ($B_{i,\{j,k\}}$) we write $B_{ij}$ ($B_{ijk}$).  Thus the vertices of the item-equivalence graph for 3 buyers are
	\[
	\begin{array}{lll}
	B_{1,\emptyset}	& B_{2,\emptyset}	& B_{3, \emptyset}	\\	
	B_{12}			& B_{21}			& B_{31}  			\\
	B_{123} 		& B_{213} 			& B_{312}			\\
	B_{13}			& B_{23}			& B_{32}	 
	\end{array}
	\]
	where each column corresponds to a different player. Note that only the non-empty sets out of these actually appear in the graph.  For 2 buyers there are at most 4 vertices in the graph:
	\[
	\begin{array}{ll}
	B_{1,\emptyset}	& B_{2,\emptyset}\\	
	B_{12}			& B_{21}
	\end{array}
	\]
	Step \ref{step:mark_3_cycles} is only relevant for the case of 3 buyers.
\end{remark}

\vspace{5pt}
\begin{algorithm}[H]
	\caption{Dynamic Pricing Scheme for up to 3 Multi-Demand Buyers.}
	\label{alg_3_players}
	
	\KwIn{Multi-demand valuations $v_1, v_2$, and also $v_3$ when $n=3$.}
	\KwOut{prices $\mathbf{p} = (p_x)_{x \in M}$.}
	
	Compute some optimal allocation $\mathbf{O}$.  \\
	Compute the preference graph $H$ and the item-equivalence graph $B$ based on $\mathbf{O}$.		\\
	Mark all edges that participate in a cycle of size 2 in $B$.\label{step:mark_2_cycles} 				\\
	In each of the cycles $B_{13}\rightarrow B_{21}\rightarrow B_{32}\rightarrow B_{13}$ and $B_{12}\rightarrow B_{31}\rightarrow B_{23} \rightarrow B_{12}$ (if these exist) choose a set of minimal size and mark its incoming edge and outgoing edge in the cycle. \label{step:mark_3_cycles} 								\\
	For every edge $B_{i_1, C_1}\rightarrow B_{i_2, C_2}$ in $B$ that was marked, and for every $x\in B_{i_1, C_1}, y\in B_{i_2, C_2}$, remove the edge $x\rightarrow y$ from $H$.  Denote the obtained graph by $H'$. \label{step:delete} \\
	Let $\Delta> 0$ be the difference in social welfare between the optimal and 2nd optimal allocation. Denote $\epsilon :=\frac{\Delta}{m + 1}$ and for every edge $e$ that was not removed (except for edges starting at the source vertex $s$) update its weight to $w'(e)=w(e)-\epsilon$. \label{epsilon}  \\
	Compute the min-weight paths from $s$ to every $x$ in $H'$, and let $\delta(s,x)$ be its weight. For every item $x$ set the price $p_x=-\delta(s,x) + \epsilon$.\label{step:prices}  \\
	\Return $(p_x)_{x \in M}$
\end{algorithm}
\vspace{5pt}


When $n=2$, the only cycle in the item-equivalence graph is $B_{12}\rightarrow B_{21}\rightarrow B_{12}$ (assuming both of these are non-empty sets), and both of its edges were marked in step $\ref{step:mark_2_cycles}$.  Thus, by Lemma \ref{item-equivalence_cycles}, all edges that participate in a 0-weight cycle in the preference graph were removed in step \ref{step:delete}.  Thus for $n=2$ Algorithm \ref{alg_3_players} is, effectively, the straightforward generalization of the Cohen-Addad {\it et al.} \cite{DBLP:conf/sigecom/Cohen-AddadEFF16} unit-demand solution to multi-demand buyers.

\vspace{0.1in}
\noindent {\bf [Running Example]}
Figure \ref{example_alg} (which can be found in Appendix \ref{proofs: section dynamic_pricing}) demonstrates the graphs $H$,$B$ and $H'$ obtained in the pricing scheme when run on our example, based on the optimal allocation $O_1=\{a,b\}, O_2= \{c,d\}, O_3 = \{e\}$.  The edges that get marked in the item-equivalence graph are $e_2$ and $e_3$ in step \ref{step:mark_3_cycles}.  This translates to the removal of the outgoing edges from $c,d$ to $e$ and from $e$ to $a,b$ when transitioning from $H$ to $H'$ and consequently no 0-weight cycles are left.  This remains true also after subtracting $\epsilon$ from every edge that does not touch $s$.
\vspace{0.1in}			

As stated before, 
computing $\mathbf{O}$, $H$ and $B$ can be done in polynomial time.  Finding the cycles in $B$ can also be done efficiently ($B$ has a constant number of vertices) as well as computing min-weight paths.  Thus the algorithm indeed runs in $poly(m)$ time as desired.

\begin{lemma} \label{positive_weights}
	After step \ref{step:delete} every cycle in $H'$ has strictly positive weight.
\end{lemma}

\begin{proof}
	Since every cycle in $H$ has non-negative weight (Corollary \ref{non_negative_cycles}), it is enough to show that at least one edge was removed from every 0-weight cycle (note that $\epsilon$ is sufficiently small enough so that any positive-weight cycle in $H$ remains positive-weight in $H'$).  By Lemma \ref{item-equivalence_cycles} it is enough to show that at least one edge was marked in each cycle in the item-equivalence graph. If $n=2$, then the only cycle in the item-equivalence graph is $B_{12}\rightarrow B_{21} \rightarrow B_{12}$, and both edges of this cycle were marked in step \ref{step:mark_2_cycles}.  Assume $n=3$ and let $C$ be a cycle in the item-equivalence graph.  Note that all outgoing edges of every vertex with 3 indices $B_{ijk}$ were marked in step \ref{step:mark_2_cycles} (for every such edge, the reverse edge also exists in the item-equivalence graph, forming a 2-edge cycle).  Furthermore, the vertices $B_{i,\emptyset}$ do not participate in any cycle as they have no incoming edges.  Thus we can assume that $C$ contains only vertices with 2 indices.  Assume w.l.o.g. that $B_{12}$ is one of the vertices in $C$.  We split to the following cases, based on the structure of $C$ starting at $B_{12}$:
	\begin{itemize}
		\item Case 1:  $B_{12} \rightarrow B_{21} \rightarrow \cdots$.  Here, the first edge is part of a 2-edge cycle and thus it was marked.
		\item Case 2:  $B_{12} \rightarrow B_{31} \rightarrow B_{13} \rightarrow \cdots$.  Here, the second edge was marked similarly to the previous bullet.
		\item Case 3:  $B_{12} \rightarrow B_{31} \rightarrow B_{23} \rightarrow \cdots$.  Here, at least one of the edges was marked in step \ref{step:mark_3_cycles}.	
	\end{itemize}	
\end{proof}

\begin{lemma} \label{prices_positive}
	For any item $x$, $p_x > 0$.
\end{lemma}

\begin{proof}
	The 0-weight path $s \rightarrow x$ is some path from $s$ to $x$, and thus $\delta(s,x) \leq 0$.  Thus $p_x = -\delta(s,x) + \epsilon > 0$ as desired.
\end{proof}

\begin{lemma} \label{strong_preference}
	For any player $i$, $x\in O_i$ and $y\not \in O_i$ ,if $e=x \rightarrow y \in H'$ then $u_i(x,\mathbf{p}) > u_i(y,\mathbf{p})$.
\end{lemma}

\begin{proof}
	By the triangle inequality we have
	\begin{alignat*}{2}
	\delta(s,x) + w'(x \rightarrow y) 					&\geq \delta(s,y)								&&\implies \\
	\delta(s,x) + v_{i}(x) - v_{i}(y) -\epsilon 		&\geq \delta(s,y)								&&\implies \\
	v_{i}(x) - (-\delta(s,x) + \epsilon) -\epsilon 		&\geq v_{i}(y) - (-\delta(s,y) + \epsilon)		&&\implies \\
	v_{i}(x) - p_x -\epsilon 							&\geq v_{i}(y) - p_y							&&\implies \\
	v_{i}(x) - p_x 										&> v_{i}(y) - p_y
	\end{alignat*}
	The claim follows.
\end{proof}

\begin{lemma} \label{utilities_positive}
	For any player $i$ and item $x \in O_i$ we have $v_{i}(x) - p_x > 0$.	
\end{lemma}

\begin{proof}
	Consider a min-weight path from $s$ to $x$ in $H'$, $s \rightarrow x_1 \rightarrow \cdots \rightarrow x_k = x$, and for every $1 \leq j \leq k$ let $i_j$ be the player such that $x_j \in O_{i_j}$ (note that $i_k = i$).  Since every cycle in $H'$ has positive weight (Lemma \ref{positive_weights}) it must be the case that all the vertices $x_i$ are different (otherwise this is not a min-weight path) and $k\leq m$. We have
	\begin{align*}
	v_{i_k}(x_k)-p_{x_k} 	&= v_{i_k}(x_k) + \delta(s,x) -\epsilon \\
	&= v_{i_k}(x_k) + \sum_{j=1}^{k-1} (v_{i_j}(x_j) - v_{i_j}(x_{j+1}) - \epsilon) - \epsilon \\
	&= \sum_{j=1}^{k} v_{i_j}(x_j) - \sum_{j=1}^{k-1} v_{i_j}(x_{j+1}) - \epsilon (k-1) -\epsilon \\  
	&\geq \mathsf{SW}(\mathbf{O}) - \mathsf{SW}(\mathbf{A}) -\epsilon m 
	\end{align*}
	
	where $\mathbf{A}$ is the allocation obtained from $\mathbf{O}$ by passing the item $x_{j+1}$ to player $i_j$ for all $j$, and dis-allocating $x_1$.  Therefore, $ \sum_{j=1}^{k} v_{i_j}(x_j) - \sum_{j=1}^{k-1} v_{i_j}(x_{j+1}) = \mathsf{SW}(\mathbf{O}) - \mathsf{SW}(\mathbf{A})$. By the assumption that every optimal allocation allocates all items, we conclude that $\mathbf{A}$ is a sub-optimal allocation and therefore the last term is positive as desired (note that $\epsilon$ is sufficiently small).
\end{proof}

We are now ready to prove that the output of our dynamic pricing scheme meets the requirements of Corollary  \ref{main_theorem}.  This is cast in the following lemma:

\begin{lemma} \label{main_goal_exposition}
	Let $\mathbf{p}$ be the price vector output by Algorithm \ref{alg_3_players}.  Then, for every player $i$ and $S \in D_{\mathbf{p}}(i)$,
	\begin{enumerate}
		\item $S$ is legal for player $i$.
		\item $S$ can be completed to a legal allocation, i.e. there exists an allocation of the items $M\setminus S$ to the other players in which every player receives a bundle that is legal for her.
	\end{enumerate}
\end{lemma}

\begin{proof}
	We prove for $i=1$ (the same proof applies also for $i=2,3$).	
	We first prove part 1. We start by showing that every $S \in D_{\mathbf{p}}(1)$ is of size $k_1$.  Since all item prices are positive (Lemma \ref{prices_positive}) and player 1 is $k_1$-demand, it cannot be the case that player 1 maximizes utility with a bundle consisting of more than $k_1$ items.  Furthermore, by Lemma \ref{utilities_positive} there are at least $k_1$ legal items from which she derives positive utility. Combining, every demanded bundle has exactly $k_1$ items. Now, for any two items $x,y$ where $x\in O_1$ and $y$ is not legal for player 1, the edge $x \rightarrow y$ was not removed in the transition from $H$ to $H'$ (since there is no corresponding edge in the item-equivalence graph that could have been marked). Thus, player 1 strongly prefers $x$ over $y$ (by Lemma \ref{strong_preference}) and we conclude that every demanded bundle contains only legal items, as desired. 
	
	We proceed to prove part 2.  Let $S \in D_{\mathbf{p}}(1)$.  We refer to the items in $S\setminus O_1$ as the items that player 1 `stole' from players 2 (and 3 if $n=3$), and to the items in $O_1\setminus S$ as those player 1 `left behind'.  We need to show that players 2 and 3 can compensate for their stolen items in a `legal manner', that is, by completing their leftover bundles $O_2\setminus S$ and $O_3\setminus S$ to $k_2$ and $k_3$ legal items, respectively.  The first step is to determine where the stolen and left behind items are taken from. 
	Since $B_{1,\emptyset}$ does not participate in any cycle in the item-equivalence graph (as it has no incoming edge), then none of its outgoing edges were marked, implying (by Lemma \ref{strong_preference}) that player 1 strongly prefers every item of $B_{1,\emptyset}$ over every item $y \notin O_1$.  Since buyer 1 derives positive utility from these items (Lemma \ref{utilities_positive}), we conclude that $B_{1,\emptyset}$ is contained in every demanded bundle and in particular in $S$.  In other words, all the items player 1 left-behind are in $B_{12}$ if $n=2$, or in $B_{12}\cup B_{13} \cup B_{123}$ if $n=3$.  Thus, if $n=2$ we are done:  buyer 2 can compensate for her stolen items by taking the leftover items in $B_{12}$ which are legal for her (the amount of stolen items equals the amount of leftover items since $\left|O_1\right| = \left|S\right| = k_1$).  We assume for the rest of the proof that $n=3$.  Since $S$ is legal for buyer 1, the stolen items $S \setminus O_1$ are contained in $B_{21}\cup B_{213} \cup B_{31} \cup B_{312}$.
	
%
%
%
	
	We denote
	\begin{align*}
	a_{2}  	&:= \left| (O_1\setminus S)\cap B_{12} 	\right| \\
	a_{3}  	&:= \left| (O_1\setminus S)\cap B_{13} 	\right| \\
	a_{23} 	&:= \left| (O_1\setminus S)\cap B_{123}	\right| \\
	b_2    	&:= \left| (S\setminus O_1)\cap B_{21} 	\right| \\
	b_{23}	&:= \left| (S\setminus O_1)\cap B_{213}	\right|	\\
	b_{3}   &:= \left| (S\setminus O_1)\cap B_{31} 	\right|	\\
	b_{32}	&:=	\left| (S\setminus O_1)\cap B_{312}	\right|
	\end{align*}
	In words, $a_{2}$ is the number of items player 1 left behind in $B_{12}$,  $b_2$ is the number of items she `stole' from player 2 out of the items in $B_{21}$, $b_{32}$ is the amount she `stole' from player 3 out of the items in $B_{312}$, etc.  By the discussion in the previous paragraph, these account for all stolen and leftover items, and we get
	\begin{equation}
	b_2 +b_{23} + b_3 + b_{32} = \left| S \setminus O_1 \right| = \left| O_1 \setminus S \right| = a_{2} + a_{23} + a_{3}. \label{bipartite_G}
	\end{equation}
	
	Consider the bipartite graph $G$ whose left side consists of the items in $S\setminus O_1$ and whose right side consists of the items in $O_1\setminus S$, with edges $(x,y)$ whenever the stolen item $x$ can be replaced by the leftover item $y$ legally (e.g., if $x\in O_2$, then $y\in B_{12} \cup B_{123}$).  Specifically, $G$ is composed of a bi-clique between the stolen items from $B_{21}\cup B_{213}$ (the stolen items of player 2) and the leftover items from $B_{12}\cup B_{123}$ (these are the leftover items that are legal for player 2), and of another bi-clique between the stolen items of $B_{31}\cup B_{312}$ (the stolen items of player 3) and the leftover items of $B_{13}\cup B_{123}$ (the leftover items that are legal for player 3).  If there is a perfect matching in $G$, then every stolen item can be replaced with the item it was matched to in the perfect matching, resulting in a legal allocation, and we are done.  Thus we assume that there is no perfect matching in $G$.  
	In this case Hall's condition does not hold for $G$. One can verify that this implies one of the following:
	\[
	b_2 + b_{23} > a_2 + a_{23} \hspace*{5pt}\text{ or }\hspace*{5pt} b_3 + b_{32} > a_3 +a_{23}
	\]
	Assume w.l.o.g. that $b_2 + b_{23} > a_2 + a_{23}$.  Then, by equation (\ref{bipartite_G}), we have $a_3 > b_3+b_{32} \geq 0$.  We claim that this implies $b_{23} = 0$.  The reason is that otherwise, player 1 stole some item, denoted $y$, from $B_{213}$ and left behind some item, denoted $x$, in $B_{13}$.  But this cannot be the case since this would imply (by Lemma \ref{strong_preference}) that the edge $x \rightarrow y$ was removed in the transition from $H$ to $H'$, but the edge $B_{13} \rightarrow B_{213}$ was never marked in the pricing scheme.  Therefore $b_{23}=0$ and $b_2 > a_2 + a_{23} \geq 0$.  The combination of $b_2 > 0$ and $a_3 >0$ implies that the edge $B_{13} \rightarrow B_{21}$ was marked in step \ref{step:mark_3_cycles}, and so one of $B_{13},B_{21}$ is of minimal size in the cycle $B_{13} \rightarrow B_{21} \rightarrow B_{32} \rightarrow B_{13}$.  In particular,
	\begin{align*}
	\left| B_{32} \right| 	&\geq \min \{\left|B_{13}\right|, \left|B_{21}\right| \} \\
	 						&\geq \min \{a_3, b_2 \} 	\\
							&\geq \min \{a_3 - \left(b_3 + b_{32}\right), b_2 - \left(a_2 + a_{23}\right)\} \\
							&= b_2 - \left(a_2 + a_{23}\right),
	\end{align*}
	where the equality holds by equation (\ref{bipartite_G}). In order to complete $S$ to a legal allocation, player 2 compensates for his stolen $b_2$ items by taking the $a_2 + a_{23}$ items player 1 left behind in $B_{12}\cup B_{123}$ and by ``stealing'' $b_2 - (a_2 + a _{23})$ items from $B_{32}$ (indeed there are enough items there for player 2 to steal).  Player 3 now has to compensate for the items stolen from her by both players, a total of 
	$(b_{32} + b_3) + (b_2 - (a_2 + a _{23})) = a_3$
	items. Since player 1 left precisely this number of items in $B_{13}$, player 3 can take them.  Note that the resulting allocation is indeed legal and thus optimal.
	
\end{proof}

\section{Maximal Domain Theorems for Walrasian Equilibrium and Dynamic Pricing} \label{sec: MDT}

In this section we prove Theorems \ref{gs_thm} and \ref{maximal_domain_theorem_dynamic_pricing}, namely the maximal domain theorems for Walrasian equilibrium and dynamic pricing, respectively.
In Appendix \ref{sec:gs_characterization} we state and prove a variant of the price based gross-substitutes characterization by Reijnierse et al. \cite{Reijnierse2002}.  As a direct corollary we obtain the following theorem, which will be needed to prove the maximal domain theorems.

\begin{theorem}\label{B_minus_A_2}
	Let $v$ be a non gross-substitutes valuation.  Then there are bundles $A,B\subseteq M$ and a price vector $\mathbf{p}$ such that
	:
		\begin{enumerate}
			\item $\left|B\setminus A\right| = 2$.
			\item $\left|A\setminus B\right| \leq 1$.
			\item $B \in \arg\min_{C~ :~ u(C,\mathbf{p}) > u(A,\mathbf{p})}(\left|(C\setminus A)\cup (A\setminus C)\right|)$.
		\end{enumerate}
\end{theorem}

\begin{proof}[Proof of Theorems \ref{gs_thm} and \ref{maximal_domain_theorem_dynamic_pricing}]
	Assume that $v_{1}$ is not gross-substitutes. Thus Theorem \ref{B_minus_A_2} implies the existence of bundles $A,B$ and a price vector $\mathbf{p}$ for which
	$\left|B\setminus A\right| = 2$, $\left|A\setminus B\right| \leq 1$, and $~B \in \arg\min_{C~ :~ u_1(C,\mathbf{p}) > u_1(A,\mathbf{p})}(\left|(C\setminus A)\cup (A\setminus C)\right|)$.
	Denote $B\setminus A = \{b_1,b_2\}$, and if $\left|A\setminus B\right| = 1$ then we denote $A\setminus B = \{a\}$.
	We now introduce our collection of unit-demand buyers.
	The first, denoted $v_{2}$, values each $b \in B\setminus A$ at $p_{b}+v_{1}(M)+1+\varepsilon'_{b}$ and values every other item at 0.
	Moreover, if $A \setminus B$ is not empty, then we have a buyer $v_{a}$	that values $a$ at $p_{a}+\varepsilon_{a}$ and values every other item at 0. 
	Similarly, we have a buyer $v_{b}$ for each item $b\in B\setminus A$
	that values $b$ at $p_{b}+\varepsilon_{b}$ and values every other item at 0.
	The values $\varepsilon_{a},\varepsilon_{b},\varepsilon'_{b}$
	are defined later. Finally, we have a buyer $v_c$ for each $c\in M\setminus (A\cup B)$ that values $c$ at $v_{1}(M)+1$ and values every other item at 0. Our goal is to set the numbers
	$\varepsilon_{a},\varepsilon_{b}, ~\varepsilon'_b$ such that the following two requirements are satisfied: if the market admits a dynamic pricing then it admits a Walrasian equilibrium, and the market does not admit a Walrasian equilibrium.
	The combination of the two requirements clearly implies both theorems. To
	this end, consider the collection $\mathcal{A}$ of all allocations
	that satisfy the following properties:
	\begin{itemize}
		\item $a$ is allocated to one of $\{ v_{1},v_{a}\} $.
		\item Each item $b\in B\setminus A$ is allocated to one of $\{ v_{1},v_{2},v_{b}\} $.
		\item Each item $c\notin A\cup B$ is allocated to $v_{c}$.
		\item Buyer $v_{2}$ takes exactly one item out of $B\setminus A$.
		\item The items in $A\cap B$ are all allocated to $v_{1}$.
	\end{itemize}
	We would like to set the numbers $\varepsilon_{a},\varepsilon_{b},\varepsilon'_{b}$
	such that no two allocations in $\mathcal{A}$ have the same social welfare.  When do two allocations
	$\mathbf{O}^{1},\mathbf{O}^{2}\in\mathcal{A}$ have the same social welfare? Consider
	the following table that specifies the difference between $\mathbf{O}^{1}$
	and $\mathbf{O}^{2}$:

	\begin{center}
		\begin{tabular}{c|c|c|c|c}
			& 1 & 2 & $a$& $b\in B\setminus A$\tabularnewline\hline 
			$\mbox{ }\mathbf{O}^{1}\mbox{ }$ & $\mbox{ }C\mbox{ }$  & $\mbox{ }e\mbox{ }$ & $\mbox{ }G\mbox{ }$& $\mbox{ }I\mbox{ }$ \tabularnewline\hline 
			$\mathbf{O}^{2}$ & $D$ & $f$ & $H$ & $J$\tabularnewline
		\end{tabular}
	\par\end{center}

	$C$ and $D$ are the bundles allocated to buyer 1 in $\mathbf{O}^{1}$ and $\mathbf{O}^{2}$, respectively. $e$ and $f$ are the items allocated to buyer 2 in
	$\mathbf{O}^{1}$ and $\mathbf{O}^{2}$, respectively. $G\subseteq \{a\}$ equals $\{a\}$ if $a$ is allocated to buyer $v_a$ in $\mathbf{O}^1$ and otherwise $G = \emptyset$. $I\subseteq B\setminus A$ is the set of items $b \in B\setminus A$ that are allocated to buyer $v_b$ in $\mathbf{O}^1$. $H,J$ are defined similarly. 
	Allocations $\mathbf{O}^{1}$ and $\mathbf{O}^{2}$ have the same social welfare exactly when
	\begin{align*}
	v_{1}(C)+(p_{e}+v_{1}(M)+1+\varepsilon'_{e})+\sum_{a\in G}(p_{a}+\varepsilon_{a})+\sum_{b\in I}(p_{b}+\varepsilon_{b}) & +\sum_{c\in M\setminus(A\cup B)}(v_{1}(M)+1)=\\
	v_{1}(D)+(p_{f}+v_{1}(M)+1+\varepsilon'_{f})+\sum_{a\in H}(p_{a}+\varepsilon_{a})+\sum_{b\in J}(p_{b}+\varepsilon_{b}) & +\sum_{c\in M\setminus(A\cup B)}(v_{1}(M)+1)
	\end{align*}
	which in turn, by rearranging, occurs exactly when
	\begin{align}\label{allocations_equal}
	\varepsilon'_{e}-\varepsilon'_{f}+\sum_{a\in G\setminus H}\varepsilon_{a}-\sum_{a\in H\setminus G}\varepsilon_{a}+\sum_{b\in I\setminus J}\varepsilon_{b}-\sum_{J\setminus I}\varepsilon_{b} & =\nonumber \\
	v_{1}(D)-v_{1}(C)+\sum_{a\in H\setminus G}p_{a}-\sum_{a\in G\setminus H}p_{a}+\sum_{b\in J\setminus I}p_{b}-\sum_{b\in I\setminus J}p_{b}+p_{f}-p_{e}
	\end{align}
	To achieve unique welfare for each allocation $O\in\mathcal{A}$ we
	must set $\varepsilon_{a},\varepsilon_{b},\varepsilon'_{b}$ so that
	equation  (\ref{allocations_equal}) never holds whenever $\mathbf{O}^{1} \neq \mathbf{O}^{2}$. The bottom expression in (\ref{allocations_equal}) is a function of $C,D,e,f,G,H,I,J$ and all of $\varepsilon_{a},\varepsilon_{b},\varepsilon'_{b}$ are in the top expression.  If we set these values so that the top expression never evaluates to 0, but also small enough so that its absolute value is always smaller than the smallest possible non-zero absolute value of the bottom expression, then equality never holds, as desired.  To this end denote the bottom expression of equation
	(\ref{allocations_equal}) by $d_{C,D,e,f,G,H,I,J}$, and define $\delta$ to be the minimal positive absolute value of $d_{{C,D,e,f,G,H,I,J}}$ among all possible choices of $C,D,e,f,G,H,I,J$.  If $\left|d_{C,D,e,f,G,H,I,J}\right| = 0$ for all possible choices, then we set $\delta = 1$.  We also define
	\[
	\varepsilon :=\min\left\{ \frac{\delta}{2},\frac{u_{1}(B,\mathbf{p})-u_{1}(A,\mathbf{p})}{4}\right\}.
	\]
	We now finally define the numbers $\varepsilon_{a},\varepsilon_{b}, \varepsilon'_{b}$ and complete the construction. We set:
		\begin{align*}
		\varepsilon_{b_{1}} & :=\varepsilon/2^{1}\\
		\varepsilon_{b_{2}} & :=\varepsilon/2^{2}\\
		\varepsilon_{a}		& :=\varepsilon/2^{3}\\
		\varepsilon'_{b_{1}}& :=\varepsilon/2^{4} \\
		\varepsilon'_{b_{2}}& :=\varepsilon/2^{5}
		\end{align*}
	We claim that whenever $\mathbf{O}^{1}$ and $\mathbf{O}^{2}$ are different allocations,
	the top expression of equation (\ref{allocations_equal}) has positive absolute value that
	is smaller than $\delta$. To see this, note that each of  $\varepsilon_{a}, \varepsilon_{b_1},\varepsilon_{b_2},\varepsilon'_{b_1},\varepsilon'_{b_2}$ appears at most once in the top expression of equation (\ref{allocations_equal}), and at least one appears whenever $\mathbf{O}^{1}$ and $\mathbf{O}^{2}$
	are not the same allocation. Take the number with the smallest power
	of 2 in the denominator and assume w.l.o.g. that it is preceeded with
	a minus sign. Then, even if the rest of the numbers appear with a plus
	sign, the entire expression still evaluates to a negative value strictly
	between -$\delta$ and 0. By definition of $\delta$ the equation
	cannot hold in this case and we obtain the desired uniqueness. We
	have proved:
	\begin{lemma}
		\label{claim:Uniqueness_in_A}For every two different allocations 
		$\mathbf{O}^{1},\mathbf{O}^{2}\in\mathcal{A}$, we have $\mathsf{SW}(\mathbf{O}^{1})\neq \mathsf{SW}(\mathbf{O}^{2})$.
	\end{lemma}
	
	\begin{corollary}
		\label{claim:Same_player}Each item $x\in M\setminus(A\cap B)$
		is allocated to the same player in every optimal allocation.
	\end{corollary}
	
	\begin{proof}
		Let $\mathbf{O}$ be an optimal allocation. The following hold with respect
		to $\mathbf{O}$:
		\begin{itemize}
			\item $a$ is allocated to one of $\{ v_{1},v_{a}\} $
			(otherwise the welfare can be increased by reallocating $a$ to $v_{a}$).
			\item Each item $b\in B\setminus A$ is allocated to one of $\{ v_{1},v_{2},v_{b}\} $
			(otherwise the welfare can be increased by reallocating $b$ to $v_{b}$).
			\item Each item $c\in M\setminus(A\cup B)$ is allocated to
			$v_{c}$ (similarly).
			\item $v_{2}$ takes exactly one item out of $M\setminus(A\cap B)$,
			and this item is in $B\setminus A$ (similarly).
		\end{itemize}
		Therefore, if we begin with the allocation $\mathbf{O}$ and reallocate all the
		items in $A\cap B$ to $v_{1}$ then the resulting allocation is in
		$\mathcal{A}$. Furthermore, since the unit-demand players value each item in $A\cap B$ at 0 and $v_1$ is monotone, we conclude that this modification does not incur a loss
		in welfare, implying that optimality is preserved. The claim follows
		since there is at most one optimal allocation in $\mathcal{A}$ (by Lemma
		\ref{claim:Uniqueness_in_A}) and the modification does not reallocate
		any item in $M\setminus(A\cap B)$.
	\end{proof}

	Corollary \ref{claim:Same_player} can be rephrased as follows:
	
	\begin{corollary}
		\label{cor:There-is-some}There is some partition of $M\setminus(A\cap B)$,
		denoted $\{ S_{1},S_{2}\} \cup\{ S_{x}\}_{x\in M\setminus(A\cap B)}$
		such that in every optimal allocation the bundle received by player
		$v_{i}$ is the union of $S_{i}$ and a subset of $A\cap B$. 
	\end{corollary}
	
	We are now ready to prove:
	\begin{lemma}
		\label{lem:Dynamic_implies_Walrasian}If the market admits a dynamic pricing, then it admits a Walrasian equilibrium.
	\end{lemma}
	
	\begin{proof}
		Let $\mathbf{q}$ be a dynamic pricing for the market (Definition~\ref{def:dynamic-pricing}). Recall that for any player $v_{i}$
		and any $S\in D_{\mathbf{q}}(v_{i})$ there is some optimal allocation
		in which $v_{i}$ receives the bundle $S$. Thus, by Corollary \ref{cor:There-is-some},
		$S_{i}\subseteq S$ for any $S\in D_{\mathbf{q}}(v_{i})$. Furthermore,
		for any player $v_{i}\neq v_{1}$ the items in $A\cap B$ do not add
		anything to the utility, implying that $S_{i}\in D_{\mathbf{q}}(v_{i})$.
		Moreover, even if we update $\mathbf{q}$ so that all items in $A\cap B$ are
		priced at 0, and denote the new price vector by $\mathbf{q'}$, then we would
		still have $S_{i}\in D_{\mathbf{q'}}(v_{i})$. Note also that this
		update can only make player $v_{1}$ want $A\cap B$ more than before.
		Thus $S_{1}\cup(A\cap B)\in D_{\mathbf{q'}}(v_{1})$. We
		have thus shown that the allocation $(S_{1}\cup(A\cap B), S_{2}, (S_{x})_{x\in M\setminus(A\cap B)})$
		together with the prices $\mathbf{q'}$ constitute a Walrasian equilibrium, as desired.
	\end{proof}
	It is left to prove that the market does not admit a Walrasian equilibrium.
	\begin{lemma}\label{requirement_2}
		The market composed of the buyers $\{ v_{1},v_{2}\} \cup\{ v_{x}\} _{x\in M\setminus(A\cap B)}$
		does not admit a Walrasian equilibrium.
	\end{lemma}
	We remark that Lemma \ref{requirement_2} proves Theorem \ref{gs_thm} and that the proof
	is mainly adapted from the original proof of Theorem \ref{gs_thm}.  
	
	\begin{proof}
		Assume towards contradiction that the allocation $Y = (Y_1, (Y_{x})_{x\in M\setminus (A\cap B)})$
		together with the price vector $\mathbf{t}$ is a Walrasian equilibrium. Let $X = (X_1, (X_{x})_{x\in M\setminus (A\cap B)})$
		be the allocation obtained from $Y$ by reallocating all of $A\cap B$
		to $v_{1}$. Define the price vector $\mathbf{q}$ as follows:
		\[
		\begin{cases}
		q_{a}=p_{a}+\varepsilon_{a} & \\
		q_{x}=0 & x\in A\cap B\\
		q_{x}=t_{x} & x\in M\setminus A
		\end{cases}
		\]
		The same arguments as in the original proof of Theorem \ref{gs_thm} show that:
		\begin{itemize}
			\item $X,\mathbf{q}$ is a Walrasian equilibrium,
			\item $A\cap B\subseteq X_{1}\subseteq A\cup B$
			\item $X_{2}=\{ b_i\} $ for one of $i=1,2$, implying
			that $B\setminus X_{1}\neq\emptyset$.  Assume w.l.o.g. that $X_2 = \{b_2\}$.
		\end{itemize}
		The last two bullets the following:
		\begin{itemize}
			\item $A\setminus X_1$ equals one of $\{a\}, \emptyset$
			\item $X_1 \setminus A$ equals one of $\{b_1\}, \emptyset$
		\end{itemize}
		and in any case, since $\left|B\setminus A\right|=2$, we have $\left|A\triangle X_{1}\right|<\left|A\triangle B\right|$. By the minimality of $B$ we have $u_{1}(X_{1},\mathbf{p})\leq u_{1}(A,\mathbf{p})<u_{1}(B,\mathbf{p})$.	Assume that $X_{1}\setminus A = \{b_1\}$.  Consider the difference $u_{1}(A,\mathbf{p})-u_{1}(X_{1},\mathbf{p}) \geq0$.
		By how much does the difference change when modifying the prices from
		$\mathbf{p}$ to $\mathbf{q}$? The items in $A\cap X_{1}$ do not contribute to the
		change (the prices of these items cancel out when evaluating the difference).
		$A\setminus X_{1}$ contributes no less than $-\varepsilon_{a}$.
		Moreover, $b_1 \notin X_{b_1}$, implying $q_{b_1}\geq p_{b_1}+\varepsilon_{b_1}$ (otherwise
		$v_{b_1}$ would prefer having $b_1$). We conclude that $X_{1}\setminus A$
		contributes at least $\varepsilon_{b_1}$ to the difference change.
		But by definition, each $b \in B\setminus A$ satisfies  $\varepsilon_{b} - \varepsilon_{a} > 0$ and in particular the difference is strictly larger at the prices $\mathbf{q}$ compared to $\mathbf{p}$.  Thus we have
		\[
		u_{1}(X_{1},\mathbf{q})<u_{1}(A,\mathbf{q})
		\]
		which is a contradiction since $(X,\mathbf{q})$ is a Walrasian
		equilibrium (implying in particular that $X_{1}$ is a favorite bundle for $v_{1}$ with respect to $\mathbf{q}$). Now assume that $X_{1}\setminus A=\emptyset$.
		Denote $d:=u_{1}(B,\mathbf{p})-u_{1}(A,\mathbf{p})>0$. When passing
		from $\mathbf{p}$ to $\mathbf{q}$, the total price of of $A\setminus X_{1}$ increased
		by at most $\varepsilon_{a}\leq \epsilon \leq d/4$ (recall the definition of $\epsilon$). Together
		with $u_{1}(X_{1},\mathbf{p})\leq u_{1}(A,\mathbf{p})$ we have
		\begin{equation}
		u_{1}(X_{1},\mathbf{q})\leq u_{1}(A,\mathbf{q})+d/4.\label{eq:-1}
		\end{equation}
		Now, since $X_{1}\subseteq A$ and $X_2 = \{b_2\}$, we must have $b_1\in X_{b_1}$, implying
		
		\[
		q_{b_1}\leq p_{b_1}+\varepsilon_{b_1}<p_{b_1}+\frac{d}{4}
		\]
		where the second inequality holds by definition of $\epsilon_{b_1}$, and the first holds since otherwise $v_{b_1}$ would rather not have $b_1$.  Moreover, since buyer $v_2$ (weakly) prefers $b_2$ over $b_1$, then we have
		\begin{align*}
		p_{b_1}+v_{1}(M)+1+\varepsilon'_{b_1}-q_{b_1} & \leq p_{b_2}+v_{1}(M)+1+\varepsilon'_{b_2}-q_{b_2} &\Leftrightarrow\\
		q_{b_2} & \leq p_{b_2}+q_{b_1}-p_{b_1}+\varepsilon'_{b_2}-\varepsilon'_{b_1} &
		\end{align*}
		
		and since $q_{b_1}\leq p_{b_1}+\varepsilon_{b_1}$, we also have
		\begin{align*}
		q_{b_2} & \leq p_{b_2}+\varepsilon_{b_1}+\varepsilon'_{b_2}-\varepsilon'_{b_1}\\
		& \leq p_{b}+\frac{d}{4}
		\end{align*}
		
		We have shown that $q_{b}\leq p_{b}+\frac{d}{4}$ for every $b\in B\setminus A$, and conclude that when passing from $\mathbf{p}$ to $\mathbf{q}$, the total price of
		$B\setminus A$ increased by at most $d/2$. Since the total
		price change of $A\setminus B$ did not decrease then we have
		\begin{equation}
		u_{1}(B,\mathbf{q})-u_{1}(A,\mathbf{q})\geq u_{1}(B,\mathbf{p})-u_{1}(A,\mathbf{p})-d/2\label{eq:-2}
		\end{equation}
		Combining (\ref{eq:-1}) and (\ref{eq:-2}) we get 
		\begin{align*}
		u_{1}(B,\mathbf{q})-u_{1}(X_{1},\mathbf{q}) & \geq u_{1}(A,\mathbf{q})+u_{1}(B,\mathbf{p})-u_{1}(A,\mathbf{p})-d/2-u_{1}(A,\mathbf{q})-d/4\\
		& =d-\frac{3}{4}d >0
		\end{align*}
		and again this is a contradiction since $(X,\mathbf{q})$ is a Walrasian
		equilibrium.
	\end{proof}
	
\end{proof}

\paragraph{Acknowledgements.} We deeply thank Amos Fiat and Renato Paes-Leme for helpful discussions.

\bibliographystyle{splncs04}
\bibliography{Dynamic_Pricing}

\appendix

\section{Missing Proofs from Section \ref{sec: dynamic_pricing}}\label{proofs: section dynamic_pricing}
\subsection{Proof of Lemma \ref{pass_alloc}}

\begin{proof}
	The weight of the cycle is
	\begin{align*}
	w(C) 	&= \sum_{i=1}^{k} w(x_i\rightarrow x_{i+1})		\\
	&= \sum_{i=1}^{k} v_{i}(x_i) - v_{i}(x_{i+1})
	\end{align*}
	Let $y$ be an arbitrary item.  If $y \notin \{x_1, \ldots x_k\}$ then it is allocated to the same player in $\mathbf{O}$ and in $\mathbf{A}$, ergo contributing 0 to the difference $\mathsf{SW}(\mathbf{O}) - \mathsf{SW}(\mathbf{A})$.  If, on the other hand, $y=x_i$, then it contributes $v_{i}(x_i) - v_{i-1}(x_i)$ to the difference (since it is allocated to player $i$ in $\mathbf{O}$ and to player $i-1$ in $\mathbf{A}$).  We have established that
	\begin{align*}
	\mathsf{SW}(\mathbf{O})-\mathsf{SW}(\mathbf{A}) &= \sum_{i=1}^{k} v_{i}(x_i) - v_{i-1}(x_i)  	\\
	&= \sum_{i=1}^{k} v_{i}(x_i) - v_{i}(x_{i+1})	\\
	&= w(C)
	\end{align*} 
\end{proof}

\subsection{Proof of Corollary \ref{non_negative_cycles}}

\begin{proof}
	If all vertices in the cycle are different then this is immediate by Lemma \ref{pass_alloc} and the fact that $\mathbf{O}$ is optimal.  If some vertices are repeated then the cycle is a union of two or more cycles with no vertex repetition (whose weights are non-negative).  The claim holds since the weight of the cycle equals the sum of the weights of the repetition-free sub-cycles.
\end{proof}

\subsection{Proof of Lemma \ref{non_negative_prices_util_dif}}

\begin{proof}
	Part 1 holds since the edge $s \rightarrow x$ is a particular path from $s$ to $x$ and its weight is 0 (and the weight of a min-weight path from $s$ to $x$ can only be smaller).  Part 2 holds by the triangle inequality:
	\begin{align*}
	\delta(s,x) + w(x \rightarrow y) &\geq \delta(s,y) 	&\implies \\
	\delta(s,x) + v_{i}(x) - v_{i}(y) &\geq \delta(s,y)		&\implies \\
	v_{i}(x) - p_x &\geq v_{i}(y) - p_y
	\end{align*}
	To show part 3, consider a min-weight path $s \rightarrow x_1 \rightarrow \cdots \rightarrow x_k = x$ from $s$ to $x$.  Let's assume first that $x_1 \notin O_i$.  Then the edge $x \rightarrow x_1$ exists and the weight of the cycle obtained by combining the edge with the path is:
	\begin{align*}
	\delta(s,x) + w(x \rightarrow x_1) = -p_x + v_{i}(x) - v_{i}(x_1) \geq 0
	\end{align*}	
	
	where the inequality holds by Corollary \ref{non_negative_cycles}, and the result follows (recall that valuations are normalized and monotone implying $v_i(x_1)\geq 0$).  We now assume that $x_1 \in O_i$. If $k=1$ (i.e., the path is simply the edge $s \rightarrow x$) then $p_x = 0$ and the result follows.  Otherwise,  $x_2 \notin O_i$ and the edge $x \rightarrow x_2$ does exist.  The weight of the cycle $x_2\rightarrow \cdots \rightarrow x \rightarrow x_2$ is
	\begin{align*}
	0 	&\leq	\delta(x_2, x) + w(x \rightarrow x_2) 					\\
	&=		\delta(x_2, x) + v_i(x) - v_i(x_2)						\\
	&=		(v_i(x_1)-v_i(x_2)) + \delta(x_2, x) + (v_i(x) - v_i(x_1)) 	\\
	&=		\delta(s,x) + v_i(x) - v_i(x_1)							\\
	&=		-p_x + v_i(x) - v_i (x_1)					
	\end{align*}
	and again the result follows.
\end{proof}

\subsection{Proof of Lemma \ref{poly_time_item-equivalence_graph}}

\begin{proof}
	Clearly the main problem is determining the non-empty sets $B_{i,C}$.  We can efficiently determine the set $B_{i,C}$ that any item $x\in O_i$ belongs to by executing the following sub-routine:  go over all players $j \in [n]\setminus \{i\}$ and compute the optimal social welfare in the residual market obtained by fixing $x$ to player $j$. Denote the result by $opt_j$ and compare $opt_j$ with the optimal social welfare in the original market, denoted by $opt$.  $x$ belongs to the set $B_{i,C}$ for the set of players $C = \{j\in [n]\setminus \{i\} \mid opt_j = opt\}$. 
\end{proof}

\subsection{Figure \ref{example_alg}}

\begin{figure}[H]
	\centering\scalebox{0.8}{
		\includegraphics[width=\linewidth]{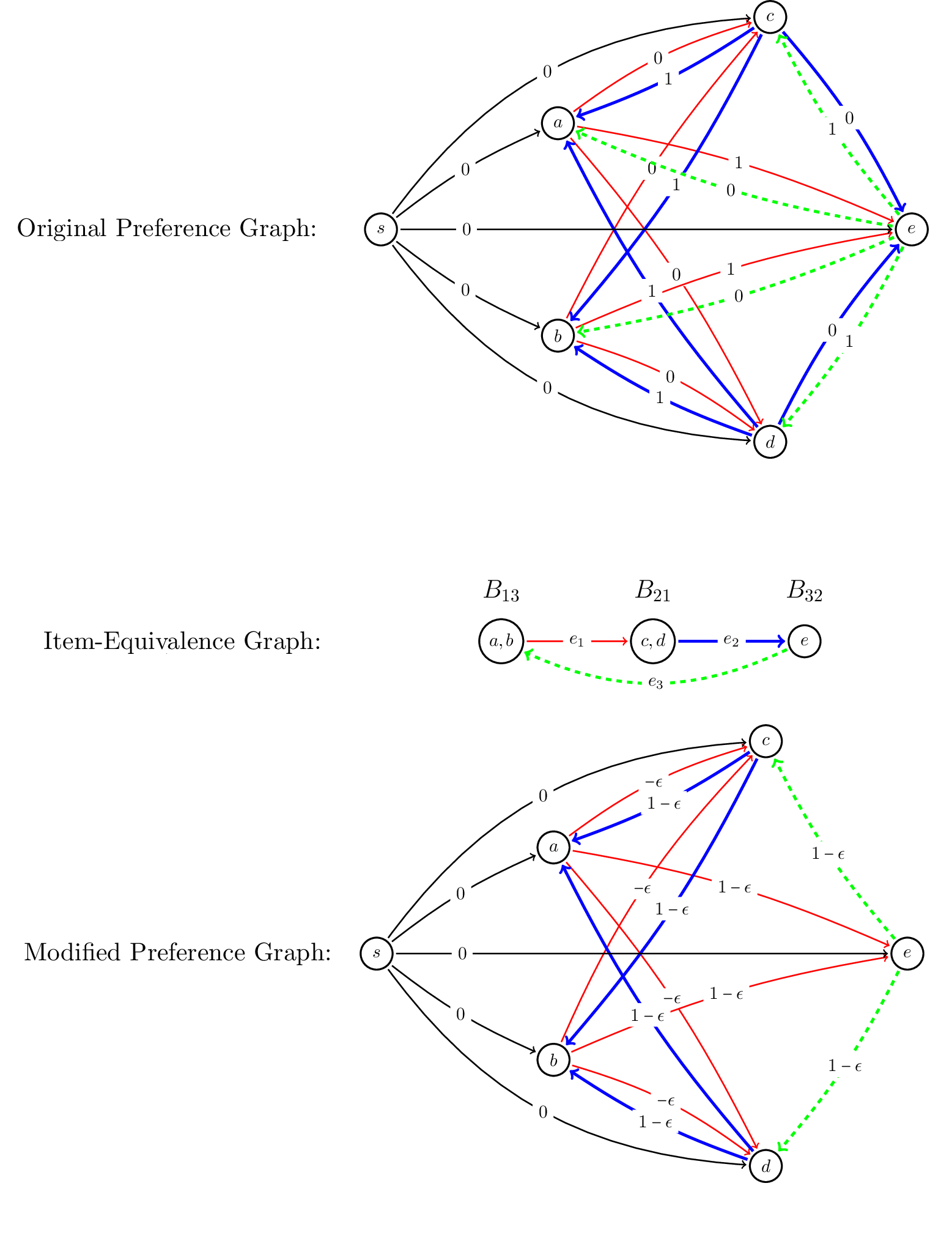}
	}
	\caption{The various graphs obtained in the execution of Algorithm \ref{alg_3_players} on our running example.
		The edges $e_2$ and $e_3$ in the item-equivalence graph were marked in step \ref{step:mark_3_cycles}.  Consequently, the edges from $c,d$ to $e$ and from $e$ to $a,b$ were removed from the original preference graph.		
	}
	\label{example_alg}
\end{figure}

\section{Dynamic Pricing for 3 Multi-Demand Buyers: the Case where Demand Exceeds Supply} \label{appendix: m < sum k_i}

In this section, we generalize the result of Section \ref{three_player_solution} for the case where $m \leq \sum_{i} k_i$. A natural approach would be to introduce imaginary items with value 0 to all players, and apply the same techniques as in Section \ref{three_player_solution} to the obtained market.  This approach ultimately succeeds, but introduces non-trivial challenges along the way which should be handled carefully.  In particular, establishing the equivalent of Lemma \ref{utilities_positive} for the generalized setting (Lemma~\ref{real_utilities_positive}) requires new ideas and more subtle arguments.

We fix a buyer profile $\mathbf{v} = (v_1,\ldots,v_n)$ over the item set $M$, where each $v_i$ is $k_i$-demand.  As explained in Section \ref{sec: dynamic_pricing} we can assume w.l.o.g. that all items are essential for optimality (i.e. each item is allocated in every optimal allocation), implying that every optimal allocation hands at most $k_i$ items to player $i$, for every $i$.  in section \ref{sec: dynamic_pricing} we made the simplifying assumption that each optimal allocation hands \textbf{exactly} $k_i$ items to player $i$ for every $i$.  This was necessary for the proof of Theorem \ref{legal_iff_optimal}, which was crucial for establishing the correctness of the dynamic pricing scheme.  In general though, the number of items might be smaller than $\sum_{i} k_i$, in which case not all players exhaust their cap $k_i$ in every optimal allocation.  In order to simulate this condition and present a dynamic pricing scheme that follows the same ideas of the scheme in the simplified setting, we introduce to the market $\sum_{i=1}^{n} k_i - m$ ``imaginary items'', valued at 0 by all players, and for every original optimal allocation in which a player $i$ receives less than $k_i$ items, we think of it as if the amount of received items is exactly $k_i$, where some of the items can be imaginary. Furthermore, it is convenient to think of the price of an imaginary item as always being 0.  We formalize these ideas in the following:

\begin{definition} $ $
	\begin{enumerate}
		\item 	The \textit{augmented market} of a buyer profile $\mathbf{v} = (v_1,\ldots,v_n)$ is the buyer profile $\mathbf{v'} = (v'_1,\ldots, v'_n)$ defined on the item set $M' = M \cup \{d_1,\ldots , d_{\sum_{i=1}^{n} k_i - m}\}$, where for every $i$, $v'_i$ is $k_i$-demand with
		\[
		v'_i(x) = 	\begin{cases}
		v_i(x) 	& x\in M\\
		0		& x \in \{d_1,\ldots , d_{\sum_{i=1}^{n} k_i - m}\}
		\end{cases}
		\]
		The items $d_1,\ldots, d_{\sum_{i=1}^{n} k_i - m}$ are called \textit{imaginary} items. 
		\item	An \textit{augmented optimal allocation} is any original optimal allocation augmented with the additional imaginary items such that every player $i$ receives exactly $k_i$ items.  Formally, an allocation $\mathbf{O}'$ is an augmented optimal allocation if every player $i$ is allocated exactly $k_i$ items, and there is an optimal allocation $\mathbf{O}$ such that for every original item $x \in M$ we have that $x \in O_i$ iff $x \in O'_i$. 
		
		\item For any price vector $\mathbf{p}$ on $M$, its \textit{augmented price vector} $\mathbf{p'}$ is the price vector on $M'$ where 
		\[
		p'_x = \begin{cases}
		p_x	& x\in M 		\\
		0	& \text{otherwise}
		\end{cases}
		\]
	\end{enumerate}	 
\end{definition}

\begin{remark}\label{augmented_optimal} $ $
	\begin{itemize}
		\item Since $v'_i$ and $v_i$ coincide on the set of real items $M$, we abuse notation and use $v_i$ when referring to $v'_i$  (and similarly for $u_i$ and $u'_i$). 
		\item Every augmented optimal allocation in the augmented market is essentially an optimal allocation in the original market, padded with the appropriate number of imaginary items, to match demand. 
		\item Since imaginary items always have value and price of 0, it follows that for any player $i$, bundle $S\subseteq M$ such that $\left|S\right| \leq k_i$ and price vector $\mathbf{p}$ on $M$ we have $S \in D_{\mathbf{p}}(i) \iff S\cup \{d_1,\ldots ,d_{k_i - \left|S\right|}\} \in D_{\mathbf{p'}}(i)$.
	\end{itemize}
\end{remark}

In our dynamic pricing scheme we adjust the tools used in the simplified setting to accommodate settings where demand exceeds supply. In particular, we still use the preference graph, except that its vertex set corresponds to all items, including imaginary ones, and its edges are defined with respect to some augmented optimal allocation $\mathbf{O}$. Analogous reasoning gives us the following lemma. Recall that $\delta(s,x)$ is the min-weight path from $s$ to $x$ in the preference graph $H$.

\begin{lemma} \label{weak_preference_real_imaginary}
	Consider the (non-negative) prices $p_x = -\delta(s,x)$ for every real item $x$, and the augmented prices $\mathbf{p'}$.  Let $i,x,y$ be such that $x \in O_i, y \notin O_i$, and both $x,y$ can be either real or imaginary.  Then player $i$ weakly prefers $x$ over $y$, i.e.
	\[
	v_i(x) - p'_x \geq v_i(y) - p'_y
	\]
\end{lemma}

An item $x \in M'$ (real or imaginary) is called \textit{legal} for player $i$ if there is some augmented optimal allocation $\mathbf{O}' = (O'_1, \ldots, O'_n)$ such that $x \in O'_i$. Note that an imaginary item is legal for some player $i$ if there is an optimal allocation in the original market in which player $i$ receives strictly less than $k_i$ items. Legal bundles and allocations are defined as in the main exposition.  Theorem \ref{legal_iff_optimal} then directly translates to our setting as follows. An allocation in the augmented market is legal iff it is \textit{augmented} optimal, implying the following lemma.
\begin{lemma} \label{main_goal_general}
	A price vector $\mathbf{p}$ is a dynamic pricing for the original market if for every player $i$ and $S \in D_{\mathbf{p}}(i)$:
	\begin{enumerate}
		\item The augmented set $S' = S \cup \{d_1,\ldots d_{\min\{k_i - \left|S\right|, \sum k_i -m\}}\}$ is legal for player $i$ in the augmented market.
		\item There exists an allocation of the items $M'\setminus S'$ to the other players in which every player receives a bundle that is legal for her.
	\end{enumerate}
\end{lemma}

The item-equivalence graph is also defined analogously. An imaginary item $x$ is in the set $B_{i,C}$ iff there is an optimal allocation in the original market in which player $i$ receives strictly less than $k_i$ items.  Furthermore, for any two imaginary items $x \in B_{i,C_1}$, $y \in B_{j,C_2}$ for $i\neq j$ we have $\{i\} \cup C_1 = \{j\} \cup C_2$ (all imaginary items are legal for the same set of players since they are valued the same by all players).  Lemma \ref{item-equivalence_cycles} carries over to our setting as well.  Equipped with the modified tools, the dynamic pricing scheme is defined analogously to the main exposition, only that prices are set only for the real items (but based on the preference graph that includes the imaginary items).  The only part of the analysis that does not carry over directly from the main exposition is the proof of Lemma \ref{utilities_positive} (ensuring a strictly positive utility from every item $x\in O_i$). This lemma was crucial to argue that each demanded set of player $i$ contains exactly $k_i$ items.  In our setting it is required in order to argue that each such demanded set contains at least $\left|O_i\cap M\right|$ items (i.e. the amount of real items in $O_i$), implying in particular that for each $S \in D_{\mathbf{p}}(i)$,
\[
\left|S'\right| = \left|S \cup \{d_1,\ldots d_{\min\{k_i - \left|S\right|, \sum k_i -m\}}\}\right| = k_i,
\] which is needed for the proof of Part 1 of Lemma \ref{main_goal_general} (the analog of Lemma \ref{main_goal_exposition}, whose proof also directly carries over to our setting).

We next explain why the proof of Lemma~\ref{utilities_positive} does not carry over to augmented markets. 
In the proof, we argued that the utility $v_i(x) - p_x$ equals $\mathsf{SW}(\mathbf{O}) - \mathsf{SW}(\mathbf{A})$, where $\mathbf{A}$ is some allocation in which some item is not allocated. Thus, $\mathbf{A}$ is necessarily sub-optimal, and the utility is strictly positive, as required. This reasoning fails in our setting since the un-allocated item might be imaginary, in which case $\mathbf{A}$ may still be optimal. In what follows we show that the lemma still holds.

\begin{lemma} \label{real_utilities_positive}
	For any player $i$ and real item $x \in O_i\cap M$, $v_{i}(x) - p_x > 0$.	
\end{lemma}

\begin{proof}
	Consider a min-weight path from $s$ to $x$ in $H'$ , $s \rightarrow x_1 \rightarrow \cdots \rightarrow x_k = x$ (recall that $H'$ is the modified preference graph), and suppose that $x_j \in O_{i_j}$ for all $j$ (with $i_k = i$).  Since every cycle in $H'$ has positive weight (Lemma \ref{positive_weights}) it must be the case that all the vertices $x_i$ are different (otherwise this is not a min-weight path).  We split to cases:
	\begin{enumerate}
		\item The item $x_1$ is real. In this case, the proof is identical to that of Lemma~\ref{utilities_positive}.
		
		\item The edge $x \rightarrow x_1$ exists (i.e., $x$ and $x_1$ do not belong to the same buyer and the edge was not removed in the transition from $H$ to $H'$).  In this case the cycle $C$ obtained by connecting $x$ to $x_1$ exists in $H'$ and its weight is positive (Lemma \ref{positive_weights}). Thus we have 
		\begin{align*}
		0	&< w(C)										\\
		&= \delta(s,x) + w(x \rightarrow x_1)		\\
		&= \delta(s,x) + v_i(x) - v_i(x_1) -\epsilon\\
		&\leq v_i(x) - (-\delta(s,x) + \epsilon)	\\
		&= v_i(x) - p_x
		\end{align*}
		as desired.

		\item One of the edges $x_d \rightarrow x_{d+1}$ of the path does not correspond to an edge in the item-equivalence graph.  I.e., if $x_d \in B_{i_d, C_d}, x_{d+1} \in B_{i_{d+1}, C_{d+1}}$, then $i_d \notin C_{d+1}$.  In this case there is no augmented optimal allocation in which $x_{d+1}$ is allocated to player $i_d$ and thus the allocation $\mathbf{A}$ obtained from $\mathbf{O}$ by passing the item $x_{j+1}$ to player $i_j$ for all $1 \leq j \leq k-1$ and dis-allocating $x_1$ is not optimal.  Here again, the proof follows from the same reasoning as the proof of Lemma~\ref{utilities_positive}.
	\end{enumerate}
	
	In the remaining cases we assume that the edge $x \rightarrow x_1$ does not exist, $x_1$ is imaginary and that all path edges correspond to item-equivalence graph edges that were not marked in the transition from $H$ to $H'$. In particular, all inner vertices of the path must belong to 2-index vertices of the item-equivalence graph (i.e. vertices of the form $B_{ij}$), since all outgoing edges from 3-index vertices in the item-equivalence graph were marked in step \ref{step:mark_2_cycles}, and 1-index sets have no incoming edges.
	
	\begin{enumerate}
		\setcounter{enumi}{3}
		\item  $x_1$ belongs to a 3-indexed set of the item-equivalence graph. Thus it cannot be an inner vertex and we have $x_1 = x_k = x$. However this is a contradiction since $x$ is a real item and $x_1$ is imaginary.
	\end{enumerate}
	
	{\bf Remark:} All cases up to now did not make any assumption on the structure of $\mathbf{O}$.  In the remaining cases we use the following terminology:  given a cycle in the item-equivalence graph in which all vertices are different, ``applying the cycle'' means choosing an arbitrary item from each vertex in the cycle, followed by returning the augmented optimal allocation obtained by reallocating each of the items to the player possessing the preceding item in the cycle. 
	
	\begin{enumerate}
		\setcounter{enumi}{4}
		\item $x_1$ belongs to a 1-indexed set.  W.l.o.g. $x_1 \in B_{1,\emptyset}$.  If $x_k\notin O_1$, then the edge $x_k \rightarrow x_1$ exists in $H'$ and this is handled in case 1.  Thus we assume that $x_k$ belongs to player 1.  $x_2 \notin B_{213}\cup B_{312}$ as otherwise $x_2$ is an inner vertex that belongs to a 3-indexed set, contradicting our assumption. Thus $x_2 \in B_{21}$ or $x_2 \in B_{31}$.  Assume $x_2 \in B_{21}$. $x_3 \notin B_{123} \cup B_{12}$ (corresponding edges were marked in Step \ref{step:mark_2_cycles}) and also $x_3 \notin B_{312}$ (cannot be inner vertex and cannot be final vertex since it does not belong to player 1).  The remaining possibility is that $x_3 \in B_{32}$.  In this case we cannot have $x_4 \in B_{13}$ (one of the edges $B_{21} \rightarrow B_{32}, B_{32} \rightarrow B_{13}$ was marked in Step \ref{step:mark_3_cycles}).  Similarly we cannot have $x_4 \in B_{23}\cup B_{213}$ (corresponding edges were marked in Step \ref{step:mark_2_cycles}).  We conclude that $x_4 = x = x_k \in B_{123}$, and the (item-equivalence graph) path is 
		\[
		B_{1,\emptyset} \rightarrow B_{21} \rightarrow B_{32} \rightarrow B_{123}
		\]
		
		In the analogous case where $x_2 \in B_{31}$, the resulting path is
		\[
		B_{1,\emptyset} \rightarrow B_{31} \rightarrow B_{23} \rightarrow B_{123}
		\]
		
		We now show that there is an alternative augmented optimal allocation in which both paths do not exist in the item-equivalence graph (i.e., one of the sets in each path is empty).  We can then update the algorithm by adding a pre-processing step in which the base augmented optimal allocation is updated to the new one if it so happens that the imaginary items belong to $B_{1,\emptyset}$.  In the new allocation the current case is vacuous.  To this end, note that the cycles  $C_1 = B_{21} \rightarrow B_{32} \rightarrow B_{123} \rightarrow B_{21}$ and $C_2 = B_{31} \rightarrow B_{23} \rightarrow B_{123} \rightarrow B_{31}$ are cycles in the item-equivalence graph. Consider the following procedure:
		
		While one of the cycles $C_1$,$C_2$ exists in the item-equivalence graph (i.e. the corresponding vertices are non-empty sets), ``apply'' one of them.
		
		Note that each application of $C_1$ decreases $\left|B_{21}\right|, \left|B_{32}\right|, \left|B_{123}\right|$ by 1 (and increases $\left|B_{23}\right|, \left|B_{312}\right|, \left|B_{12}\right|$ by 1).  Each application of $C_2$ decreases $\left|B_{31}\right|, \left|B_{23}\right|, \left|B_{123}\right|$ by 1 (and increases $\left|B_{32}\right|, \left|B_{213}\right|, \left|B_{13}\right|$ by 1). In particular, the sum
		\[
		\left|B_{21}\right| + \left|B_{123}\right| + \left|B_{31}\right|
		\]
		strictly decreases after each iteration.  We conclude that the procedure must end and none of these paths exists in the obtained augmented optimal allocation.

		
		\item $x_1$ is imaginary and belongs to a 2-indexed set of the item-equivalence graph.  W.l.o.g. $x_1 \in B_{32}$.  It cannot be the case that $x_2 \in B_{23}\cup B_{213}$ (the corresponding edges were marked in Step \ref{step:mark_2_cycles}).  If $x_2 \in B_{123}$, then $x_2 = x$ but the edge $x_2 \rightarrow x_1$ was not removed in the transition to $H'$, and this case was covered in part 2. The remaining possibility is that $x_2 \in B_{13}$.  Note that the edge $x_2 \rightarrow x_1$ exists in $H'$ and thus $x_2 \neq x$.  It cannot be the case that $x_3 \in B_{21}$ (one of the edges $B_{32} \rightarrow B_{13}, B_{13} \rightarrow B_{21}$ was marked in Step \ref{step:mark_3_cycles}).  $x_3 \notin B_{312}\cup B_{31}$ (the corresponding edges were marked in Step \ref{step:mark_2_cycles}). The last possibility is that $x_3 = x \in B_{213}$ and the corresponding item-equivalence graph path is 
		\[
		B_{32} \rightarrow B_{13} \rightarrow B_{213}
		\]
		
		(indeed, the edge $x_3 \rightarrow x_1$ does not exist in $H'$ since the edge $B_{213} \rightarrow B_{32}$ was marked in Step \ref{step:mark_2_cycles}). Furthermore, every imaginary item can belong either to $B_{32}$ or $B_{23}$. In the analogous case where $x_1 \in B_{23}$, the resulting item-equivalence graph path is
		\[
		B_{23} \rightarrow B_{12} \rightarrow B_{312}
		\]
		As in the previous case we will show that there is a ``cycle application'' procedure that results in an alternative augmented optimal allocation in which none of these paths exists. Denote the cycles
		\begin{align*}
		C_1 &= B_{32} \rightarrow B_{13} \rightarrow B_{213} \rightarrow B_{32} \\
		C_2 &= B_{23} \rightarrow B_{12} \rightarrow B_{312} \rightarrow B_{23}
		\end{align*}
		and our assumption is that $C_1$ exists in the bottlneck graph.  The procedure goes as follows:
		\begin{itemize}
			\item While $\left|B_{13}\right|, \left|B_{213}\right|, \left|B_{12}\right|, \left|B_{312}\right| \geq 1 $ (all inequalities hold)
			\begin{itemize}
				\item Apply $C_1$, then apply $C_2$
			\end{itemize}
			\item If $\left|B_{12}\right| = 0$ or $\left|B_{312}\right| = 0$, apply $C_1$ $\min\{\left|B_{13}\right|, \left|B_{213}\right|, \left|B_{32}\right|\}$ times and terminate.
			\item Otherwise (i.e., $\left|B_{13}\right| = 0$ or $\left|B_{213}\right| = 0$), apply $C_2$ $\min\{\left|B_{12}\right|, \left|B_{312}\right|,\left|B_{23}\right|\}$ times and terminate.
		\end{itemize}
		
		Each application of $C_1$ decreases $\left|B_{13}\right|, \left|B_{213}\right|, \left|B_{32}\right|$ by 1, and increases $\left|B_{123}\right|, \left|B_{23}\right|, \left|B_{31}\right|$ by 1.  Each application of $C_2$ decreases $\left|B_{12}\right|, \left|B_{312}\right|,\left|B_{23}\right|$ by 1 and increases $\left|B_{123}\right|, \left|B_{32}\right|,\left|B_{21}\right|$ by 1.  Therefore, in total, each iteration of the loop decreases each of $\left|B_{13}\right|, \left|B_{213}\right|, \left|B_{12}\right|, \left|B_{312}\right|$ by 1 and the loop in the process ends, with either $C_1$ or $C_2$ non-existent. The final step takes care of eliminating the other cycle (note that applying the leftover cycle does not resurrect its counterpart cycle).
	\end{enumerate}
\end{proof}

\section{A Characterization of Gross-Substitutes}\label{sec:gs_characterization}

In their paper, Reijnierse et al. prove the following:

\begin{theorem}[\cite{Reijnierse2002}]\label{GS_iff_SM_RGP}
	A valuation $v$ is gross-substitutes if and only if the following two conditions hold:
	\begin{itemize}
		\item For every pair of different items $x,y$ and bundle $S\subseteq M\setminus \{x,y\}$,	we have
		\begin{equation}\label{SM}
		v\left(S\cup\left\{ x\right\} \right)+v\left(S\cup\left\{ y\right\} \right)\geq v\left(S\right)+v\left(S\cup\left\{ x,y\right\} \right)\tag{SM}
		\end{equation}
		
		\item For every triplet of different items $x,y,z$ and bundle $S\subseteq M\setminus \{x,y,z\}$, we have
		\begin{equation}\label{RGP}
		v\left(S\cup\left\{ x\right\} \right)+v\left(S\cup\left\{ y,z\right\} \right)\leq\max\left\{ v\left(S\cup\left\{ y\right\} \right)+v\left(S\cup\left\{ x,z\right\} \right),v\left(S\cup\left\{ z\right\} \right)+v\left(S\cup\left\{ x,y\right\} \right)\right\}\tag{RGP}
		\end{equation}
	\end{itemize}
\end{theorem}

The first condition is the well-known submodularity condition.  The conditions \eqref{SM} and \eqref{RGP} have analogous ``price'' counterparts:

\begin{lemma}[\cite{Reijnierse2002}]\label{SM_RGP_iff_P_SM_P_RGP}
	A valuation $v$ satisfies \eqref{SM} and \eqref{RGP} if and only if it satisfies the following two conditions:
	\begin{itemize}
		\item There are no vector $\mathbf{p} \in \mathbb{R}^{\left|M\right|}$ (possibly with negative entries), two different items $x,y$ and a bundle $S\subseteq M\setminus \{x,y\}$ for which
		\begin{equation}\label{P-SM}
		D_{\mathbf{p}}(v) = \{S,S\cup\{x,y\}\}\tag{P-SM}
		\end{equation}
		
		\item There are no vector $\mathbf{p} \in \mathbb{R}^{\left|M\right|}$ (possibly with negative entries), three different items $x,y,z$	and a bundle $S\subseteq M\setminus \{x,y,z\}$ for which
		\begin{equation}\label{P-RGP}	
		D_{\mathbf{p}}(v) = \{S\cup\{x\},S\cup\{y,z\}\}\tag{P-RGP}
		\end{equation}
	\end{itemize}
\end{lemma}

The combination of Theorem \ref{GS_iff_SM_RGP} and Lemma \ref{SM_RGP_iff_P_SM_P_RGP} (essentially Lemma 4.2 in \cite{Leme17}) implies that if $v$ is not gross-substitutes, then either \eqref{P-SM} or \eqref{P-RGP} are violated.  If, for example, \eqref{P-SM} is violated, then there are a vector $\mathbf{p}$, two different items $x,y$ and a bundle $S\subseteq M\setminus \{x,y\}$ such that 
\[
D_{\mathbf{p}}(v) = \{S,S\cup\{x,y\}\}.
\]

We can then decrease $p_x$ and $p_y$ by a small enough amount so that $S\cup\{x,y\}$ becomes the unique utility-maximizing bundle, and $S$ becomes the only 2nd best bundle.  It would appear that taking the vector $\mathbf{p}$ together with $A := S$ and $B = S\cup\{x,y\}$ proves Theorem \ref{B_minus_A_2}.  However, the prices obtained from Lemma \ref{SM_RGP_iff_P_SM_P_RGP} can be negative (and indeed are in the known construction) and therefore are unsuitable.  The same problem arises when assuming that \eqref{P-RGP} is violated.

The following is a different version of Lemma \ref{SM_RGP_iff_P_SM_P_RGP} with non-negative prices.

\begin{lemma}\label{SM_RGP_iff_NP_SM_NP_RGP}
	A valuation $v$ satisfies \eqref{SM} and \eqref{RGP} if and only if it satisfies the following two conditions:
	\begin{itemize}
		\item There are no nonnegative price vector $\mathbf{p}$, two different items $x,y$
		and a bundle $S\subseteq M\setminus \{ x,y\}$ for which	$p_{x},p_{y}>0$ and
		\begin{equation}\label{NP-SM}
		\{S,S\cup \{ x,y\}\} \subseteq D_{\mathbf{p}}(v) \subseteq \{ T \mid T\subseteq S\}\cup \{T\cup \{ x,y\} \mid T\subseteq S\}\tag{NP-SM}
		\end{equation}
		
		\item There are no nonnegative price vector $\mathbf{p}$, three different items $x,y,z$ and a bundle $S \subseteq M\setminus \{x,y,z\}$ for which $p_x,p_y,p_z > 0$ and 
		\begin{equation}\label{NP-RGP}
		\{S\cup \{x\}, S\cup \{y,z\}\} \subseteq D_{\mathbf{p}}(v) \subseteq \{T \cup \{x\} \mid T \subseteq S\} \cup \{T \cup \{y,z\} \mid T\subseteq S\}\tag{NP-RGP}
		\end{equation}
	\end{itemize}
\end{lemma}

The combination of Theorem \ref{GS_iff_SM_RGP} and Lemma \ref{SM_RGP_iff_NP_SM_NP_RGP} imply:
\begin{theorem} \label{GS_iff_NP_SM_NP_RGP}
	A valuation $v$ is gross-substitutes iff it satisfies \eqref{NP-SM} and \eqref{NP-RGP}.
\end{theorem}

We now show how Theorem \ref{GS_iff_NP_SM_NP_RGP} implies Theorem \ref{B_minus_A_2}.  The proof of Lemma \ref{SM_RGP_iff_NP_SM_NP_RGP}, which to a large extent is adapted from \cite{Leme17} and \cite{Roughgarden14}, is given right after. 

\begin{proof}[Proof of Theorem \ref{B_minus_A_2}]
	Let $v$ be valuation that is not gross substitutes.  By Theorem \ref{GS_iff_NP_SM_NP_RGP}, there is a nonnegative price vector $\mathbf{p}$ and a bundle $S$ such that one of the following holds:
	\begin{enumerate}
		\item There are items $x,y\notin S$ for which $p_{x},p_{y}>0$ and
		\[
		\{S,S\cup \{ x,y\}\} \subseteq D_{\mathbf{p}}(v) \subseteq \{ T \mid T\subseteq S\}\cup \{T\cup \{ x,y\} \mid T\subseteq S\}
		\]
		\item There are items $x,y,z\notin S$ for which $p_{x},p_{y},p_{z}>0$ and 
		\[
		\{S\cup \{x\}, S\cup \{y,z\}\} \subseteq D_{\mathbf{p}}(v) \subseteq \{T \cup \{x\} \mid T \subseteq S\} \cup \{T \cup \{y,z\} \mid T\subseteq S\}
		\]
	\end{enumerate}
	Assume 1, and decrease the price of $x$ and $y$ by a small enough $\epsilon>0$ so that the new prices are still nonnegative and $S$ derives the 2nd highest utility under the new prices.  Observe that all utility maximizing bundles under the updated prices contain $x,y$, and $S\cup \{x,y\}$ is such a bundle.  Thus, if we choose $A := S$ and  $B := S\cup\{x,y\}$, then $A,B$ satisfy $\left|B\setminus A\right|=2$, $\left|A\setminus B\right| = 0$, and every other utility-maximizing bundle $C$ satisfies $\left|C\triangle A\right| \geq \left|B\triangle A\right|$, as required.  Likewise, if bullet 2 holds, then we can take $A:=S\cup\{x\}$,  $B=S\cup \{ y,z\}$ together with the price vector $\mathbf{p}$ after $p_y$ and $p_z$ have been decreased by a small enough amount.
\end{proof}

\begin{proof}[Proof of Lemma \ref{SM_RGP_iff_NP_SM_NP_RGP}]
	We first show that \eqref{SM} is equivalent to \eqref{NP-SM}, and then we show that under the assumption that \eqref{SM} holds, \eqref{RGP} is equivalent to \eqref{NP-RGP}.  The combination implies the lemma.
	
	Assume that \eqref{NP-SM} does not hold, i.e., there are corresponding non-negative price vector $\mathbf{p}$,	items $x,y$ and a bundle $S$. Then,
	\begin{align*}
	S,S\cup \{x,y\}  			& \in D_{\mathbf{p}}(v)	\\
	S\cup \{x\} ,S\cup \{y\}	& \notin D_{\mathbf{p}}(v)
	\end{align*}
	implying
	\[
	v(S)+v(S\cup \{x,y\}) - 2\cdot \sum_{d\in S}p_{d}-p_{x}-p_{y} > v(S\cup \{x\} )+v(S\cup \{ y\}) -2\cdot \sum_{d\in S}p_{d}-p_{x}-p_{y}.
	\]
	We conclude that \eqref{SM} is violated.  For the converse direction, assume \eqref{SM} is violated; i.e.,
	\[
	v(S) + v(S\cup \{x,y\}) = v(S\cup \{x\}) + v(S\cup \{y\})+\delta
	\]
	for some $\delta>0$. We define the prices $\mathbf{p}$ as follows.
	Set $p_{d}=\infty$ for any $d\notin S\cup \{x,y\}$ to
	guarantee that no item outside of $S\cup \{x,y\}$ is
	demanded. Set $p_{d}=0$ for any $d\in S$. Finally, let $\epsilon := \delta/2$ and set
	\begin{align*}
	p_{x} & :=v(x | S\cup \{y\}) - \epsilon		\\
	p_{y} & :=v(y | S\cup \{x\}) - \epsilon
	\end{align*}
	First we show that $p_{x},p_{y}$ are positive:
	\begin{align*}
	p_{x} 	& := v(S\cup \{x,y\}) - v(S\cup \{y\}) - \frac{1}{2}(v(S)+v(S\cup \{x,y\}) - v(S\cup\{ x\}) - v(S\cup \{y\}))	\\
	& =\frac{1}{2}( v(S\cup\{ x,y\}) - v(S\cup\{ y\}) + v(S\cup\{ x\}) -v(S))	\\
	& \geq \frac{1}{2}( v(S\cup\{ x,y\}) - v(S\cup\{ y\})) \\
	& \geq \frac{1}{2}\delta \\
	& >0,
	\end{align*}
	where the first and second inequalities hold since $v(S\cup\{ x\}) -v(S) \geq 0$ ($v$ is monotone).  $p_{y}>0$ can be
	shown similarly. The definitions of $p_{x}$ and $p_{y}$ directly imply
	\[
	u(S\cup\{ x,y\}) > u(S\cup\{ x\} ) , u(S\cup\{ y\})
	\]
	and it is also immediate that
	\[
	p_{x}+p_{y} = v(S\cup\{ x,y\} )-v(S)
	\]
	implying
	\[
	u(S\cup\{ x,y\} )=u(S).
	\]
	To summarize, 
	\begin{equation}
	u(S),u(S\cup\{ x,y\} )>u(S\cup\{ x\} ),u(S\cup\{ y\} )\label{eq:cond}
	\end{equation}
	Finally, for any $T\subseteq S$ and $E\subseteq\{ x,y\} $
	we have (recall that $p_{d}=0$ for all $d\in S$ and that $v$ is monotone):
	\begin{align*}
	u(T\cup E) & =v(T\cup E)-\sum_{d\in E}p_{d}\\
	& \leq v(S\cup E)-\sum_{d\in E}p_{d}\\
	& =u(S\cup E)
	\end{align*}
	and this together with the inequalities \eqref{eq:cond} imply that $S,S\cup\{ x,y\} $
	are demanded and any other demanded bundle must be of the form $T$
	or $T\cup\{ x,y\} $ for $T\subseteq S$. Thus \eqref{NP-SM} is violated, as required.
	
	We proceed to prove that under the assumption that $v$ satisfies \eqref{SM}, Conditions \eqref{RGP} and \eqref{NP-RGP} are equivalent.
	
	Assume that \eqref{NP-RGP} does not hold, i.e.,there are corresponding non-negative price vector $\mathbf{p}$, items $x,y,z$ and a bundle $S$. Then
	
	\begin{align*}
	S\cup\{ x\} ,S\cup\{ y,z\}  & \in D_{\mathbf{p}}(v)\\
	S\cup\{ x,z\} ,S\cup\{ x,y\} ,S\cup\{ y\} ,S\cup\{ z\}  & \notin D_{\mathbf{p}}(v)
	\end{align*}
	In particular we have
	\begin{align*}
	v(S\cup\{ x\} )-p_{x}+v(S\cup\{ y,z\} )-p_{y}-p_{z} & >v(S\cup\{ y\} )-p_{y}+v(S\cup\{ x,z\} )-p_{x}-p_{z}\\
	v(S\cup\{ x\} )-p_{x}+v(S\cup\{ y,z\} )-p_{y}-p_{z} & >v(S\cup\{ z\} )-p_{z}+v(S\cup\{ x,y\} )-p_{x}-p_{y}
	\end{align*}
	
	implying
	
	\begin{align}
	v(S\cup\{ x\} )+v(S\cup\{ y,z\} ) & >v(S\cup\{ y\} )+v(S\cup\{ x,z\} )\label{rgp_violation1} \\
	v(S\cup\{ x\} )+v(S\cup\{ y,z\} ) & >v(S\cup\{ z\} )+v(S\cup\{ x,y\} )\label{rgp_violation2}
	\end{align}				
	and thus \eqref{RGP} is violated as required.
	
	For the converse direction, assume that (\ref{rgp_violation1}) and (\ref{rgp_violation2}) hold.  We define the prices $\mathbf{p}$ as follows. Set $p_{d}=\infty$
	for any $d\notin S\cup\{ x,y,z\} $ to guarantee that no
	item outside of $S\cup\{ x,y,z\} $ is demanded. By rearranging
	the assumed inequalities we get
	\begin{align*}
	v(z\mid S\cup\{ y\} ) & >v(z\mid S\cup\{ x\} ) \geq 0\\
	v(y\mid S\cup\{ z\} ) & >v(y\mid S\cup\{ x\} ) \geq 0	
	\end{align*}
	Therefore, by setting $p_{y}=v(y\mid S\cup\{ z\} )-\epsilon$
	and $p_{z}=v(z\mid S\cup\{ y\} )-\epsilon$
	for a sufficiently small $\epsilon$, $p_{y}$ and $p_{z}$
	are positive and 
	\begin{equation}
	u(S\cup\{ y,z\} ) > u(S\cup\{ y\} ),u\{ S\cup\{ z\} \} > u(S), \label{eq_util1}
	\end{equation}
	
	where the first inequality is a direct implication of the definition
	of $p_{y}$ and $p_{z}$ and the second inequality is implied by the first since $v$ is submodular.  Next we define $p_{x}$ so that the utility from $S\cup\{ x\} $
	equals that of $S\cup\{ y,z\} $. Specifically:
	\begin{align*}
	p_{x} & =v(S\cup\{ x\} )-v(S\cup\{ y,z\} )+p_{y}+p_{z}\\
	& =v(S\cup\{ x\} )-v(S\cup\{ y,z\} )+v(S\cup\{ y,z\} )-v(S\cup\{ z\} )-\epsilon+v(S\cup\{ y,z\} )-v(S\cup\{ y\} )-\epsilon\\
	& =v(S\cup\{ x\} )+v(S\cup\{ y,z\} ) - v(S\cup\{ z\} )-v(S\cup\{ y\} )-2\epsilon \\
	& >v(S\cup\{ y\} )+v(S\cup\{ x,z\} ) - v(S\cup\{ z\} )-v(S\cup\{ y\} )-2\epsilon \\
	& \geq v(S\cup \{x\}) - v(S\cup \{z\}) -2\epsilon \\
	& >0,
	\end{align*}
	where the first inequaity holds by (\ref{rgp_violation1}), the second holds by monotonicity and the third holds by (\ref{rgp_violation2}) and for a small enough $\epsilon$.  Finally, set $p_{d}=0$ for every
	$d\in S$. Observe that indeed all prices are nonnegative as required.
	Next we show that
	\begin{align}
	u(S\cup\{ x\} ) & >u(S\cup\{ x,y\} ) \label{eq_util_2}\\
	u(S\cup\{ x\} ) & >u(S\cup\{ x,z\} ) \label{eq_util_3}
	\end{align}
	Adding the two inequalites together and applying the submodularity of $v$ gives 
	\begin{align}
	u(S\cup\{ x\} )>u(S\cup\{ x,y,z\} ) \label{eq_util_4}
	\end{align}

	We show that (\ref{eq_util_2}) holds ((\ref{eq_util_3}) is analogous). This amounts to showing that the
	marginal contribution of $y$ to the bundle $S\cup\{ x\} $
	is negative:
	\begin{align*}
	v(y|S\cup\{ x\} )-p_{y} & =v(S\cup\{ x,y\} )-v(S\cup\{ x\} )-p_{y}\\
	& =v(S\cup\{ x,y\} )-v(S\cup\{ x\} )-(v(S\cup\{ y,z\} )-v(S\cup\{ z\} )-\epsilon)\\
	& =-(v(S\cup\{ y,z\} )+v(S\cup\{ x\} )-v(S\cup\{ x,y\} )-v(S\cup\{ z\} ))+\epsilon\\
	& <0
	\end{align*}
	where the inequality holds by (\ref{rgp_violation2}) for a sufficiently small $\epsilon>0$.  To summarize, the combination of (\ref{eq_util1}),(\ref{eq_util_2}), (\ref{eq_util_3}), (\ref{eq_util_4}) together with $u(S\cup\{x\}) = u(S\cup \{y,z\})$ establishes that each of
	\begin{align*}
	&u(S\cup\{ x\} ), \\
	&u(S\cup\{ y,z\} ) 
	\end{align*}
	is strictly greater than each of
	\begin{align*}
	&u(S),\\
	&u(S\cup\{ y\} ),\\
	&u(S\cup\{ z\} ),\\
	&u(S\cup\{ x,z\} ),\\
	&u(S\cup\{ x,y\} ),\\
	&u(S\cup\{ x,y,z\} ).
	\end{align*}
	Now, for any $T\subseteq S$ and any $E\subseteq\{ x,y,z\} $
	we have (recall that $p_{d}=0$ for all $d\in S$)
	\begin{align*}
	u(T\cup E) & =v(T\cup E)-\sum_{d\in E}p_{d}\\
	& \leq v(S\cup E)-\sum_{d\in E}p_{d}\\
	& =u(S\cup E)
	\end{align*}
	Therefore, demanded bundles can only be of the form $T\cup\{ x\} $
	or $T\cup\{ y,z\} $ for $T\subseteq S$, and $S\cup\{ x\} ,S\cup\{ y,z\} $
	are demanded, implying that \eqref{NP-RGP} is violated. This concludes the proof.
\end{proof}

\section{A Budget-Additive Market with Dynamic Pricing and without Walrasian Equilibrium}\label{appendix:example_dp_no_we}

Consider the example given by \cite{FL14}, with 4 players $c_i,d_i$ for $i \in \{1,2\}$, and 7 items $a_i,b_i,\alpha_i$ for $i \in \{1,2\}$ and $\beta$.  Each of the buyers is budget-additive with budget 2, meaning that the value for a bundle $S$ is $v(S) = \min\{2,\sum_{x \in S}v(\{x\})\}$.  For $i \in \{1,2\}$, $c_i$ has value 1 for $a_i,b_i,\alpha_i$ and value 0 for the rest of the items, and $d_i$ has value 2 for $\beta$, value 1 for $a_i,b_i$ and value 0 for the rest.  It can be easily verified that $OPT = 7$ in this market.  Furthermore, \cite{FL14} showed that the market does not admit a Walrasian equilibrium.  However, we claim that it does admit an optimal dynamic pricing.  To see this, consider the first round prices 
\begin{align*}
p^1_{\alpha_1} = p^1_{\alpha_2} = p^1_{\beta} 	&= \epsilon 	\\
p^1_{a_1} = p^1_{a_2} 							&= 2\epsilon	\\
p^1_{b_1} = p^1_{b_2}							&= 3\epsilon
\end{align*}
We split to cases according to the first incoming buyer, and we can assume w.l.o.g. that it is either $c_1$ or $d_1$ (the other cases are symmetric):
\begin{enumerate}
	\item Buyer $c_1$ arrives first.  Under the above prices he takes $\alpha_1$ and $a_1$.  At this point, we set the following prices for all subsequent rounds:
	\begin{align*}
	p^2_{\beta} = p^2_{b_1}	&= \epsilon 	\\
	p^2_{a_2} = p^2_{\alpha_2} 	&= 2\epsilon	\\
	p^2_{b_2}					&= 3\epsilon	\\
	\end{align*}
	
	The earlier of $d_1,d_2$ to arrive takes $\beta$.  If $d_1$ arrives later he takes $b_1$ and $c_2$ takes $a_2$ and $\alpha_2$ .  If $d_2$ arrives later and before $c_2$ he takes $a_2,b_2$ and $c_2$ takes $\alpha_2$.  If $d_2$ arrives last, $c_2$ takes $a_2$ and $\alpha_2$ and $d_2$ takes $b_2$.  $OPT$ is achieved in any alternative.
	
	\item Buyer $d_1$ arrives first. In this case, the initial prices are used throughout. Buyer $d_1$ takes $\beta$. $c_1$ takes $a_1,\alpha_1$ regardless of his place in line.  If $d_2$ arrives before $c_2$, he takes $a_2,b_2$ and $c_2$ takes $\alpha_2$.  Otherwise, $c_2$ takes $a_2,\alpha_2$ and $d_2$ takes $b_2$.  Again, $OPT$ is achieved.  
\end{enumerate}

\end{document}